\documentclass[11pt,letterpaper]{article}
\pdfoutput=1

\usepackage[margin=1in]{geometry}

\title{Searching in trees with monotonic query times\footnote{Some results of this work have been presented at the 47th International Symposium on Mathematical Foundations of Computer Science (MFCS 2022) \cite{mfcs-DereniowskiW22}.} \footnote{Partially supported by National Science Centre (Poland) grant number 2018/31/B/ST6/00820.}}

\usepackage{authblk}
\author[1]{\Large Dariusz Dereniowski}
\author[1,2]{\Large Izajasz Wrosz}
\affil[1]{\small Faculty of Electronics, Telecommunications and Informatics, Gda\'{n}sk~University~of~Technology,~Poland}
\affil[2]{\small Intel,~Poland}
\date{}

\usepackage{amssymb,amsmath}
\usepackage{enumitem}
\usepackage{graphicx}
\usepackage{hyperref}

\usepackage{amsthm}
\theoremstyle{definition}
\newtheorem{definition}{Definition}[section]
\theoremstyle{plain}
\newtheorem{lemma}[definition]{Lemma}
\newtheorem{theorem}[definition]{Theorem}

\newtheorem{observation}[definition]{Observation}
\newtheorem{corollary}[definition]{Corollary}

\usepackage{clrscode3e}
\usepackage[noend, noline, ruled, linesnumbered]{algorithm2e}
\SetAlFnt{\normalsize}
\DontPrintSemicolon
\SetKwComment{Comment}{$\triangleright$\ }{}
\SetKwProg{Def}{def}{:}{}
\SetArgSty{textnormal}
\SetKwRepeat{Do}{do}{while}
\SetKwComment{Comment}{$\triangleright$\ }{}

\newcommand{\visibilityN}[1]{\textup{scr}(#1)}
\newcommand{\visibilityT}[1]{T_{\textup{scr}}(#1)}

\newcommand{\tempU}{u^*}
\newcommand{\tempV}{v^*}
\newcommand{\tempX}{z}
\newcommand{\targetv}{t}
\newcommand{\NP}{\textup{NP}}
\newcommand{\cI}{\mathcal{I}}
\newcommand{\bigo}{\mathcal{O}}
\newcommand{\nat}{\mathbb{N}}

\providecommand{\costFunc}{\omega}
\providecommand{\strategy}{\mathcal{A}}
\providecommand{\costOfStrategyAt}[1]{COST(\strategy,#1)}
\providecommand{\costOfStrategy}{COST(\strategy)}
\providecommand{\OPT}[1]{OPT(#1)}

\providecommand{\strategySequence}[1] {\mathcal{Q}_{\strategy}(T,#1)}
\providecommand{\strategySequenceAt}[2]{\mathcal{Q}_{\strategy,#2}(T,#1)}

\providecommand{\COST}[1]{COST(#1)}
\providecommand{\COSTweights}[2]{COST(#1,#2)}

\providecommand{\inputTree}[2]{#1 = (V, E, #2) }
\providecommand{\rootNode}{r}
\providecommand{\neighbors}[1]{N(#1)}

\providecommand{\decTree}[1]{\mathcal{#1}}
\providecommand{\decTreeFull}[1]{\mathcal{#1} = (\mathcal{V}, \mathcal{E})}
\newcommand{\cV}{\mathcal{V}}
\providecommand{\corrNodes}[2]{\mathcal{#1}(#2)}

\begin{document}

\maketitle

\begin{abstract}
We consider the following generalization of binary search in sorted arrays to tree domains. In each step of the search, an algorithm is querying a vertex $q$, and as a reply, it receives an answer, which either states that $q$ is the desired target, or it gives the neighbor of $q$ that is closer to the target than $q$. A further generalization assumes that a vertex-weight function $\costFunc$ gives the query costs, i.e., the cost of querying $q$ is $\costFunc(q)$. The goal is to find an adaptive search strategy requiring the minimum cost in the worst case. This problem is NP-complete for general weight functions and one of the challenging open questions is whether there exists a polynomial-time constant factor approximation algorithm for an arbitrary tree? In this work, we prove that there exist a constant-factor approximation algorithm for trees with a monotonic cost function, i.e., when the tree has a vertex $v$ such that the weights of the subsequent vertices on the path from $v$ to any leaf give a monotonic (non-increasing or non-decreasing) sequence $S$. This gives a constant factor approximation algorithm for trees with cost functions such that each such sequence $S$ has a fixed number of monotonic segments. Finally, we combine several earlier results to show that the problem is NP-complete when the number of monotonic segments in $S$ is at least $4$.

\end{abstract}

\section{Introduction}

Searching plays a central role in computer science due to its ubiquitous real-world applications as well as its inherent relation to other problems. Many fundamental computational problems, e.g. sorting, can be formulated as searching for an element in an appropriately defined set or apply searching as a subroutine. See for example the classic book \cite{Knuth} or a recent work on relation between sorting and binary search trees \cite{BlellochD23}. The search problem, both in the real world as well as in abstract terms, typically consists of a search space (e.g., a set, a house, a graph), a search target (vertex, or car keys) and a searcher that follows some search strategy. In the engineering of data management systems, searching information efficiently is essential for achieving high performance. Problem-solving by searching is a well established approach in artificial intelligence \cite{AIModernApproach}. In operations research, algorithmic techniques applied to searching significantly improved efficiency of searching for submarines \cite{Hohzaki2016}, to give one example. On the other hand, the works on the naval applications contributed to formulation of numerous theoretical models of search games, e.g., the binary search game.

An important aspect of any search models is the amount of information available to the search agent about the potential location of the search target \cite{AIModernApproach}. The \emph{uninformed search} model is such that no additional information is available beyond the definition of the search space. The depth-first- or breadth-first-search and related models are examples of uninformed search algorithms. On the other hand, the \emph{informed search} models assume that some additional information is available to the search agent. The additional information can be exploited to construct significantly faster search strategies. In practical setting, the additional information comes from exploiting the context of the search problem, e.g., an assumption that alike elements according to some similarity measure are \emph{close} to each other, another practical example is an assumption that elements of the search space are sorted. In a theoretical setting, the additional information can be modeled as an heuristic function or an oracle.

In the category of the informed search models we can further characterize the search models in terms of \emph{when} the additional information is exploited. I.e., we can first calculate a search strategy, e.g., in a form of a decision tree, based on the definition of the search space, which later would be used to locate arbitrary elements of the search space. Using the terms of supervised machine learning, we first learn the search model, which is then being inferred when searching for particular search targets. Consequently, the optimization criteria may be the time to learn the strategy, the accuracy and the time of locating particular search targets. In a contrary online approach, the search strategy is being refined at the same time as locating particular search targets.

\subsection{Problem statement} \label{sec:problem-statement}

We consider the following search problem defined for a node-weighted tree $\inputTree{T}{\costFunc}$, with a query cost function $\costFunc\colon V(T)\rightarrow \mathbb{R_+}$. The goal is to design an adaptive algorithm that we call a \emph{search strategy} or just \emph{strategy} for short.
The strategy is divided into steps and in each step it performs a \emph{query} by selecting some vertex $q$ of $T$.
As an \emph{answer} it receives information that either $q$ is the \emph{target} or otherwise it learns which neighbor of $q$ is closer to the target.\footnote{Alternatively, one may think of the latter as learning which component of $T-q$ contains the target.}
In the former case, the search is completed.
The \emph{cost} of the query is $\costFunc(q)$.

To formally define the search strategy, consider a node weighted tree $\inputTree{T}{\costFunc}$ and a search target $\targetv\in V(T)$. If $\vert V(T)\vert = 1$, then search strategy $\strategy(T, \targetv)$ is the element $\targetv$ itself. Otherwise, let $q\in V(T)$ be the vertex that the strategy queries first when searching in $T$, let $C_{i}$ be the $i$-th connected component of $T\setminus\{q\}$. Then, for $\vert V(T)\vert > 1$, $\strategy(T, \targetv)$ is a tuple $(q, \strategy(C_{1}, t), ..., \strategy(C_{k}, \targetv) )$, where $k=\vert \neighbors{q}\vert$ is the number of neighbors of $q$. We note that this kind of search definition can be understood as a form of a decision tree.

We denote by $\costOfStrategyAt{t}$ the sum of costs of all queries performed by a strategy $\strategy$, when finding $\targetv$.
The \emph{cost} of $\strategy$ is defined as $\costOfStrategy = \max_{\targetv\in V}\costOfStrategyAt{\targetv}$.
The optimization criterion asks for a strategy of minimum cost.

One might think about this search process that with each query the search is narrowed to a subtree of $T$ that can still contain the target.
The process is \emph{adaptive} in the sense that the choice of the subsequent elements to be queried depends on the locations and answers of the previous queries.
Also, we consider only \emph{deterministic} algorithms, where given a fixed sequence of queries performed so far (and the information obtained), the algorithm selects always the same element as the next query for a given tree. E.g., the optimal algorithm for the classic binary search problem on paths that always queries the median element is adaptive and deterministic. It can be equivalently stated that the algorithm calculates a \emph{strategy} for the given input tree, which encodes the queries to be performed for each potential target.

Our problem is a generalization of the classical binary search, where one considers trees instead of paths. While also the generalized problem can be found in the literature under the name of binary search,
it might be illustrative to note that answers to queries over a tree domain are not necessarily binary (\textit{left} or \textit{right}), as the set of possible answers depends on the degree of the queried vertex.

\begin{figure}
	\centering
	\includegraphics[width=1.0\textwidth]{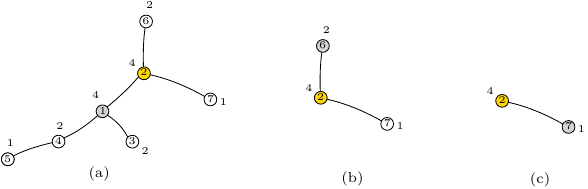}
	\caption{Binary search on weighted trees. The input tree (a) contains a target vertex (2), whose position is unknown to the algorithm. Three queries are performed to the vertices 1, 6, and 7, which incurs the costs of 4, 2 and 1, respectively.
    The subtrees that contain all vertices that are still candidates for the target after the first two queries are shown in (b) and (c).}

	\label{fig:graph_search}
\end{figure}

\subsection{Paper Organization}
The remaining parts of the paper are organized as follows. Section~\ref{sec:motivation} explains the motivation for our work. We describe practical applications of the studied search problem and several use cases for the monotonic structure of the cost function. Section~\ref{sec:earlier-techniques} provides a summary of the techniques and approaches used previously to tackle the problem of binary search in trees. Section~\ref{sec:our-results} describes the main results of this work, which is followed by Section~\ref{sec:related-work} that surveys the related work. Subsequently, in Section~\ref{sec:preliminaries} we introduce the notation used in the paper as well as several essential concepts that are commonly used in the related body of research. In Sections~\ref{sec:structured_trees}~-~\ref{sec:aligned-strategies}, we develop several concepts and techniques (i.e., structured decision trees, aligned strategies and bottom-up processing). These techniques are later used in Sections~\ref{sec:algorithm-up-monotonic}~-~\ref{sec:algorithm-general-monotonic} where we formulate and analyze algorithms for searching in trees with different types of monotonic cost functions. Section~\ref{sec:conclusions} concludes the paper, where we propose final remarks and suggest potential directions for future research.

\subsection{Motivation and applications} \label{sec:motivation}

We study a generalization of one of the most fundamental and well known problems in computer science, i.e., binary search. Since graphs form a natural abstraction for processes in many domains, our generalization into the domain of graphs allows us to seek for potential applications of our algorithms and related problems for example in the domain of large networks analysis \cite{future-big-graphs}. Compared to unordered data structures, it is known that maintaining at least a partial order over the elements improves the performance of fundamental operations like search, update or insert, in terms of the number of comparisons needed \cite{Knuth}. In the context of data management systems, it is known that efficient indexing is crucial to obtain good performance. Some of the indexing schemes have the form of a tree, for example the GIndex. At the same time, calculating optimal search strategies for structures with a partial order is hard in general \cite{trees-node-weighted-hard-DereniowskiN06}. Practical approximation algorithms for searching in large graphs are important for the advancement of systems focused on large graph analysis.

Binary search on trees seen as a graph ranking problem (cf. \cite{Dereniowski08}) can be used to model parallel Cholesky factorization of matrices \cite{application-cholesky-DereniowskiK03}, scheduling of parallel database queries \cite{MakinoUI01}, and VLSI layouts \cite{application-VLSI-SenDG92}. A more general search model, i.e., binary search on graphs has been used as a framework for interactive learning of classifiers, clusterings or rankings \cite{Emamjomeh-ZadehK17,application-learning-Emamjomeh-Zadeh20}.
More applications can be found in software testing \cite{Ben-Asher} or mobile agent computing \cite{BoczkowskiFKR21}.

There are several practical applications of searching where the cost of a query is naturally \emph{higher} in some central vertex and decreases while moving towards the leaf nodes of the tree.
This occurs naturally when considering data access times in computer systems, e.g., due to memory hierarchies of modern processors, characteristics of the storage devices, or distributed nature of data management systems. In the literature on the binary search problem, a Hierarchical Memory Model of computation has been studied \cite{app-monotonic-HMM}, in which the memory access time is monotonic with respect to the location of a data element in the array. Another work was devoted to modeling the problem of text retrieval from magnetic or optical disks, where a cost model was such that the cost of a query was monotonic from the location of the previous query performed in the search process \cite{paths-NavarroBBZC00}.
We also mention an application of tree domains in automated bug search in computer code \cite{Ben-Asher}.
In such case naturally occurs a possibility that the tree to be searched has the monotonicity property.
In particular, each query represents performing an automatic test that determines whether the part of the code that corresponds to the subtree under the tested vertex has an error.
Thus, the vertices that are closer to the root represent larger parts of the code and thus may require more or longer lasting tests. The scale of exa-scale software applications that may consist of several tens of separate software packages motivates the use of advanced software validation strategies, in order to quickly identify software regressions, e.g., when updating versions of several different packages.

An open question has been posed several times about whether a constant-factor approximation algorithm exists for arbitrary weight functions, see e.g. \cite{tree-open-question-Angelidakis,tree-bicriteria-BorowieckiDO21,tree-binary_indentification-CicaleseJLV11}.
The main motivation of this work is to address this question by finding natural input instances for which a constant-factor approximation is achievable.
In particular, we give an algorithm for down-monotonic cost functions and motivated by the theory of parameterized complexity, study a parameter $k$ of the input instance: the maximum number of monotonic segments on any path from a central node to any leaf. The case of $k$ equal to the tree height captures the class of general cost functions.
We see this as a step towards solving the general problem. Our techniques lead to a constant approximation for any fixed $k$. We believe that this method can be improved to allow combining segments of different monotonicity, while preserving the constant factor approximation.

\subsection{Earlier techniques} \label{sec:earlier-techniques}
Historically, the transition from paths to trees had an intermediate step where one wants to search partial orders instead of linear ones \cite{edge-rank-linear-LamY98,trees-MozesOnakW08O}.
Performing comparisons in partial orders leads to a similarly defined graph problem except that one queries edges instead of vertices.
However, each tree instance for querying edges can be solved via finding an optimal search strategy for a tree instance with querying vertices \cite{DereniowskiKUZ17}.
For these reasons, we consider the more general vertex-query version on trees.

It turns out that our problem reappeared several times in theoretical computer science under different names.
It has been called \emph{LIFO-search} \cite{GiannopoulouHT12} in the area of graph pursuit-evasion games, \emph{ordered coloring} \cite{KatchalskiMS95}, minimum height elimination trees \cite{BodlaenderGKH91} in the context of parallel matrix factorization, then vertex ranking \cite{node-ranking-IyerRV88}, and more recently as \emph{tree-depth} \cite{NesetrilM06}.

While listing the main approaches to construct search strategies, we do not differentiate between the edge and vertex versions of the problem as it turns out that the algorithmic techniques we mention are applicable to both.
The majority of works, especially for the unweighted version of the problem utilize the bottom-up tree processing.
In this approach one usually assigns greedily the lowest possible interval and argues that this leads to an optimal solution, see e.g. \cite{Ben-Asher,edge-ranking-IyerRV91,edge-rank-linear-LamY98,trees-MozesOnakW08O,Onak}.
We note here that by `interval' we mean the following interpretation of the cost function: if costs are seen as query durations, then the query interval is the corresponding time period in which the query is processed.
The particular contributions of these works mainly lie in various optimizations that bring down the computational complexity of finding optimal strategies.
In case of non-uniform query times the problem becomes more challenging and keeping in mind that it is NP-complete \cite{tree-edge-weighted-hard-Dereniowski06} a series of works focused on finding approximate solutions.
(We note that the unweighted case has been long known to be linear-time solvable \cite{node-rank-linear-Schaffer89}).
A natural approach, for the case of non-uniform costs, is to query a `central' element of the search space and then computing the strategy recursively for the remaining subtrees \cite{tree-edge-weighted-hard-Dereniowski06,graphs-Emamjomeh-Zadeh16}, which for trees gives $\bigo(\log n)$-factor approximations.
More involved recursive approaches allowed to break the barrier of $\bigo(\log n)$ in \cite{tree-binary_indentification-CicaleseJLV11,tree-weighted-CicaleseKLPV16} giving a $\bigo(\log n/\log\log n)$-approximation.
A rather involved bottom-up approach gave the best approximation ratio of $\bigo(\sqrt{\log n})$ so far \cite{DereniowskiKUZ17}.
The main open question in this line of research is whether either the edge or vertex query model admits a polynomial-time constant factor approximation for weighted trees.
A possible angle at tackling this question is to consider either special subclasses of trees \cite{tree-binary_indentification-CicaleseJLV11} or restricted weight functions \cite{mfcs-DereniowskiW22}.

\subsection{Our results} \label{sec:our-results}

Given a tree $\inputTree{T}{\costFunc}$, we say that a cost function $\costFunc$ is \emph{up-monotonic} (respectively \emph{down-monotonic}) if there exists a vertex $r\in V$ such that for any $u,v \in V$, if $v$ lies on the path between $r$ and $u$ in $T$, then $\costFunc(u)\leq \costFunc(v)$ ($\costFunc(v)\leq \costFunc(u)$, respectively). We say that $\costFunc$ is \emph{monotonic} if it is either up-monotonic or down-monotonic. For this kind of $\omega$ we get the following algorithms.

\begin{theorem} \label{thm:up-monotonic}
There exists a linear-time 8-approximation adaptive search algorithm for an arbitrary tree with up-monotonic cost function.
\end{theorem}

\begin{theorem} \label{thm:down-monotonic}
There exists a polynomial-time $2$-approximation adaptive search algorithm for an arbitrary tree with down-monotonic cost function.
\end{theorem}
We say that $\costFunc$ is \emph{$k$-monotonic}, $k\geq 1$, if $T$ can be partitioned into vertex disjoint subtrees $T_1,\ldots,T_l$ such that each of them is either up-monotonic or down-monotonic, and for some $r\in V$, any path from $r$ to a leaf in $T$ intersects at most $k$ trees $T_i$, $i\in\{1,\ldots,l\}$.
By combining Theorem~\ref{thm:up-monotonic} and Theorem~\ref{thm:down-monotonic}, we obtain:
\begin{theorem} \label{thm:k-monotonic}
There exists a polynomial-time $8k$-approximation adaptive search algorithm for an arbitrary tree with a $k$-monotonic cost function.
\end{theorem}
This implies a constant factor approximation algorithm for instances with fixed $k$ and raises a question whether the problem is FPT with respect to $k$?
We give a negative answer to this question by proving:
\begin{theorem} \label{thm:npc}
The problem of finding a vertex search strategy for bounded diameter spiders with $k$-monotonic cost functions is strongly $\NP$-complete for any fixed $k\geq 4$.
\end{theorem}
A \emph{spider} is a tree with at most one vertex of degree greater than two.

\subsection{Related work} \label{sec:related-work}

One can distinguish three different research directions related to our problem.
We list main results separately for each of those as follows.

\medskip
\noindent
\textbf{Vertex search in trees.} It is known that an optimal vertex search strategy can be computed in linear time for an unweighted tree and this result has been obtained independently in \cite{node-rank-linear-Schaffer89} and in \cite{Onak}.
On the other hand, the weighted version becomes (strongly) NP-complete and FPT with respect to the maximum cost\footnote{With an assumption of integer-valued cost functions.} \cite{trees-node-weighted-hard-DereniowskiN06}.
For weighted paths there exists a dynamic programming $\bigo(n^2)$-time optimal algorithm \cite{tree-binary_indentification-CicaleseJLV11} (see also \cite{LaberMP02}), which although covers graphs of special structure, is interesting as it may be viewed as the classical binary search with non-uniform comparison costs.
For general trees there is a polynomial-time $\bigo(\sqrt{\log n})$-approximation algorithm~\cite{DereniowskiKUZ17}.
The latter is derived from a QPTAS that for any $0<\varepsilon<1$ gives a $(1+\varepsilon)$-approximation in time $n^{\bigo(\log n/\varepsilon^2)}$ \cite{DereniowskiKUZ17}.
We also remark on a closely related problem of finding, for an arbitrary input graph, its spanning tree that has the minimum search cost \cite{MiyataMNZ06}.
Another generalization that appeared in the context of parallel scheduling problems, expressed in our terminology, is to find a search strategy in which one is allowed to perform several queries simultaneously \cite{ZhouNN95}.
This line of research can be placed in a wider context of analyzing adaptive algorithms versus some coding schemes and we refer interested reader e.g. to a much broader survey in~\cite{Pelc02}.
Another related model of searching in trees is where the weight of each vertex indicates the likelihood of the vertex being a search target. I.e., this search model can be considered as a generalization of finding a Binary Search Tree for a path to the domain of trees \cite{BerendsohnK22,LaberM11}.

\medskip
\noindent
\textbf{Edge search in trees.}
The edge search problem on unweighted trees is essentially different than the vertex search.
There has been a series of papers providing polynomial-time algorithms which eventually lead, again independently under the names of edge ranking and searching in partial orders, to linear-time optimal algorithms \cite{edge-rank-linear-LamY98,trees-MozesOnakW08O}.
Moreover, although this gives the same complexity as for vertex search, the algorithms themselves were much more involved and thus the edge search version is more challenging to analyze for unweighted trees.
Despite the fact that the complexity is resolved, upper bounds on the number of queries have been also developed \cite{application-db-queries-DereniowskiK06,graphs-Emamjomeh-Zadeh16,MakinoUI01}.
We briefly refer to similar generalizations including spanning tree computations \cite{MakinoUI01} or the `batch' querying \cite{ZhouKN96}.
Interestingly, there is a simple edge-subdividing reduction making the weighted vertex search more general than the weighted edge search~\cite{DereniowskiKUZ17}.
However, there is no such transition from non-weighted edge search to non-weighted vertex search or vice versa.
The weighted edge search is NP-complete for several already restricted classes of trees \cite{tree-weighted-CicaleseKLPV16,tree-edge-weighted-hard-Dereniowski06}.

\medskip
\noindent
\textbf{Vertex search in general graphs.} We note that our vertex search problem has a natural reformulation for graphs: a query to a vertex $q$, whenever $q$ is not the target gives \emph{any} neighbor of $q$ that is closer to the target than $q$.
This version has been introduced in~\cite{graphs-Emamjomeh-Zadeh16}.
It turns out that $\log_2 n$ queries are enough to perform a search in any graph which gives a quasi-polynomial algorithm \cite{graphs-Emamjomeh-Zadeh16}.
The same work proves that the weighted version is PSPACE-complete.
See also for further references e.g. \cite{DeligkasMS19,graph-median-DereniowskiLU21,graphs-DereniowskiTUW19,Emamjomeh-ZadehK17}.

\section{Definitions and Preliminaries} \label{sec:preliminaries}

The set of vertices and edges of a tree $T$ is denoted by $V(T)$ and $E(T)$, respectively. For a connected subset of vertices $V'\subseteq V(T)$ of $T$, we denote by $T[V']$ the induced subtree of $T$ consisting of all vertices from $V'$. For a rooted $T$, we denote its root by $\rootNode(T)$, while for any $v\in V(T)$, $T_v$ is the subtree of $T$ rooted at $v$ which consists of $v$ and all its descendants. $\neighbors{v}$ denotes the set of neighbors of a vertex $v$.

For any two intervals $I=[a,b)$ and $I'=[a',b')$ we write $I>I'$ ($I<I'$) when $a\geq b'$ ($b\leq a'$). Thus, whenever $I>I'$ or $I<I'$, we have $I\cap I'=\emptyset$. For an interval $I=[a,b)$ and a number $c$, we write $I<c$ ($I>c$), when $b\leq c$ ($a\geq c$, respectively).
Furthermore, we define the sum $I+c$ as the interval $[a+c,b+c)$.

We say that an interval $I=[a,b)$ is \emph{minimal} (\emph{maximal}) in a given set of intervals, if and only if for any other interval $I'=[a',b')$ in the set, we have that $a\leq a'$ (respectively $b\geq b'$). Note that every finite set of intervals contains both a minimal and a maximal interval. In this work, all intervals we use are left-closed, right-opened, and with integer endpoints.

We extend the classic lexicographic order over sequences of integers, denoted by binary operator $<_{l}$ to sequences of intervals.
For two interval sequences $S$ and $S'$ (each with pairwise disjoint intervals), we write $S<_l S'$ if and only if there exists an integer $i$ such that $i+1\notin\bigcup S$, $i+1\in\bigcup S'$ and
\[[0,i)\cap\bigcup S = [0,i)\cap\bigcup S'.\]
We then write $S\leq_l S'$ if and only if $S<_l S'$ or $S=S'$.

We use the term \emph{sequence of queries} to refer to the sequence of vertices queried by a search strategy $\strategy$ for a target $t$ and we denote it by $\strategySequence{t}$. The vertex queried by $\strategy$ in the $i$-th step is denoted by $\strategySequenceAt{t}{i}$. Given a search strategy for $T$, one can easily generate a search strategy for any subtree $T'$ by simply discarding the queries to the vertices outside of $T'$.
A useful way of encoding search strategies is by using decision trees.
\begin{definition} \label{def:decision_tree}
	We define a \emph{decision tree} for $\inputTree{T}{w}$ generated by a search strategy $\strategy$ as a rooted tree $\decTreeFull{T}$, where $\cV = V(T)$. The root of $\decTree{T}$ is the first vertex queried by $\strategy$. For any $v\in \cV$, $v$ has a child $v'$ if and only if $v'$ is queried by $\strategy$ right after $v$ for some choice of the target. Each $v\in \cV$ corresponds to a subset $\corrNodes{T}{v}\subset V(T)$ which contains all vertices that can still contain the target when $\strategy$ queries $v$.
\end{definition}

It can be seen that the root $r$ of any decision tree corresponds to $\corrNodes{T}{r}=V(T)$ and the leaves correspond to single vertices in $V(T)$.

Lemma~\ref{lem:component-in-dec-tree} shows a property of the decision trees with respect to the positions of vertices that belong to some connected component of the input tree.

\begin{lemma} \label{lem:component-in-dec-tree}
	Let $\decTree{T}$ be a decision tree for $T$. If $T'$ is a subtree of $T$, then there exists a vertex $u\in V(T')$, such that for every $v\in V(T')$ it holds $v\in V(\decTree{T}_{u})$.
\end{lemma}
\begin{proof}
	We traverse the decision tree $\decTree{T}$ starting at the root of $\decTree{T}$ and following a path to some leaf. If $\rootNode(\decTree{T})\in V(T')$, then the lemma follows. Otherwise, consider children $v_{i}$ of $\rootNode(\decTree{T})$. Because in $T$ there is only one edge incident to $\rootNode(\decTree{T})$ that lies on a path to all vertices of $T'$, all vertices of $T'$ belong to exactly one subtree $\decTree{T}_{v_{i}}$. We continue the traversal in $\decTree{T}_{v_{i}}$ and set $u$ as the first visited vertex in this process that belongs to $T'$.
\end{proof}

We now formally define the cost of a search process, which is the parameter of $T$ the algorithm aims to minimize. In terms of the search strategy $\strategy$ for an input tree $\inputTree{T}{\costFunc}$, the sum of costs of all queries performed when finding a target $t\in V(T)$ is denoted by $\costOfStrategyAt{t}$, while the \emph{cost} of $\strategy$ is defined as $\costOfStrategy = \max_{t\in V}\costOfStrategyAt{t}$.

Sometimes in our analysis, we will take one decision tree and measure its cost with different query times, i.e., for different weights.
This approach will allow us to artificially increase the cost of some queries in the decision tree, which we formalize as follows.
Let $\decTree{T}$ be a decision tree for $T=(V,E,\costFunc)$ and $\costFunc'$ an arbitrary cost function defined on $V(T)$, potentially a different one than $\costFunc$. We denote by $\COSTweights{\decTree{T}}{\costFunc'}$ the \emph{cost of the decision tree according to the cost function} $\costFunc'$:
	\[\COSTweights{\decTree{T}}{\costFunc'}=\max\big\{\sum_{v\in V(P)} \costFunc'(v) \;|\; P\textup{ is a path from the root to a leaf in }\decTree{T}\big\}.\]
The \emph{cost} of $\decTree{T}$ is then $\COSTweights{\decTree{T}}{\costFunc}$ and we shorten it to $\COST{\decTree{T}}$. (Note that one may equivalently define the cost of $\decTree{T}$ by taking $\COST{\decTree{T}}=\costOfStrategy$ where $\decTree{T}$ is generated by $\strategy$.)
The minimum cost that is achievable for $T$ is denoted by
	\[\OPT{T}=\min\{\COST{\decTree{T}} \;|\; \decTree{T}\textup{ is a decision tree for }T\}.\]
We say that $\decTree{T}$ \emph{is optimal} if and only if $\COST{\decTree{T}} = \OPT{T}$.

We will use the following simple folklore lemma whose proof is given for completeness.
Lemma~\ref{lem:dec-tree-for-subtree} has been previously formulated in \cite{tree-binary_indentification-CicaleseJLV11} in the context of a slightly different problem of edge search in weighted trees.
\begin{lemma} \label{lem:dec-tree-for-subtree}
	If $\decTree{T}$ is a decision tree for $T$, then for any subtree $T'$ of $T$ there exists a decision tree $\decTree{T'}$, such that $\COST{\decTree{T'}} \leq \COST{\decTree{T}}$.
\end{lemma}

\begin{proof}
	We first apply Lemma~\ref{lem:component-in-dec-tree} and reduce $\decTree{T}$ to its subtree rooted at $u\in V(T')$.
	Then, we sequentially remove from $\decTree{T}_{u}$ all nodes that do not belong to $T'$. For every removed node $v$, consider all subtrees rooted at the children of $v$ in the current decision tree. When removing $v$ we also remove all subtrees that do not contain any vertex from $T'$ but we connect the roots of all other subtrees (note that there will be only one such subtree) directly under the parent of $v$. Hence, the reduced tree is still a decision tree. What is more, all paths from the root to a leaf in the obtained decision tree result from corresponding paths from $\decTree{T}$ by removing zero or more nodes, which upperbounds the cost of $\decTree{T'}$ by the cost of $\decTree{T}$.
\end{proof}

We say that a cost function $\costFunc$ is \emph{rounded} if the values taken by $\costFunc$ are powers of two (for any $v\in V(T)$, $\costFunc(v)=2^k$ for some $k\in \mathbb{N}$). We also say that $T'$ is the \emph{rounding} of $T$ if it is obtained from $T$ by rounding the cost function to the closest, greater power of two. We obtain that $\OPT{T'} \leq 2\,\OPT{T}$ because $\costFunc' \leq 2\costFunc$. On the other hand, Lemma~\ref{lem:2rounded_dec+tree_cost} shows how large is the overhead from applying an optimal decision tree for the input with a rounded cost function to search through the input tree with the original weights.

\begin{lemma} \label{lem:2rounded_dec+tree_cost}
	Let $T'=(V,E,\costFunc')$ be the rounding of $T=(V,E,\costFunc)$.
	Let $\decTree{T}$ and $\decTree{T'}$ be optimal decision trees respectively for $T$ and $T'$.
	We have:
	\[\COSTweights{\decTree{T'}}{\costFunc} \leq 2\,\COST{\decTree{T}} = 2\OPT{T}. \]
\end{lemma}

\begin{proof}
	We have $\COSTweights{\decTree{T'}}{\costFunc} \leq \COSTweights{\decTree{T'}}{\costFunc'} \leq \COSTweights{\decTree{T}}{\costFunc'} \leq 2\COSTweights{\decTree{T}}{\costFunc}$.	
	The first inequality holds because $\costFunc \leq \costFunc'$. The second inequality holds because $\decTree{T'}$ is optimal for $T'$. The last inequality holds because $\costFunc' \leq 2\costFunc$.
\end{proof}

\subsection{Strategy functions}

In case of uniform query times, an optimal search strategy can be obtained by calculating a \emph{strategy function} $f\colon V(T) \to \mathbb{N}$, where for each pair of distinct vertices $v_1, v_2 \in V(T)$, if $f(v_1) = f(v_2)$, then each path connecting $v_1$ and $v_2$ has a vertex $v_3$ such that $f(v_3) > f(v_1)$. As outlined earlier, methods using functions analogous to $f$ can be found in the literature on problems such as elimination trees, vertex ranking \cite{edge-rank-linear-LamY98,node-rank-linear-Schaffer89}, tree-depth~\cite{NesetrilM06}, and others. One property of the strategy function is that it encodes the (reversed) order of the queries. E.g., in any sequence of queries, a vertex with a higher value of $f$ is always queried before those with lower values of the strategy function. Thus, the maximal value of $f$ over the vertices of a tree is the worst-case number of queries necessary to search the tree. In order to express the variable costs of queries we extend the strategy function into an \emph{extended strategy function} (see also~\cite{tree-edge-weighted-hard-Dereniowski06}).
\begin{definition} \label{def:extended_strategy_function}
	A function $f\colon V(T) \to \{[a,b)\colon0\leq a<b\}$, where $\rvert f(v)\rvert \geq\costFunc(v)$ for each $v\in V(T)$, is an \emph{extended strategy function} if for each pair of distinct vertices $v_1, v_2 \in V(T)$, if $f(v_1)\cap f(v_2) \neq \emptyset$, then the path connecting $v_1$ and $v_2$ has a $v_3$ such that $f(v_3)> f(v_1)\cup f(v_2)$.
\end{definition}
We say that an $f$ is \emph{optimal} for $T=(V,E,\costFunc)$ if and only if $\sup\bigcup_{v\in V}f(v)$ is minimal among all extended strategy functions for $T$.\footnote{We note that $\bigcup_{v\in V}f(v)$ is an interval with the left endpoint equal to $0$ for an optimal $f$.}
Based on the time interpretation of strategy functions, optimal functions correspond to optimal strategies.

\begin{observation} \label{obs:optimal-strategy-function}
    For any $T=(V,E,\costFunc)$, there exists an optimal extended strategy function $f$, such that $\vert f(v)\vert=\costFunc(v)$, for every $v\in V(T)$.
    \qed
\end{observation}
\begin{proof}
    Let $f'$ be an optimal extended strategy function for $T=(V,E,\costFunc)$, while $v\in V(T)$, an arbitrary vertex such that $\vert f'(v)\vert >\costFunc(v)$. Taking $f'(v)=[a_v, b_v)$ and $\delta_v=\vert f'(v)\vert - \costFunc(v)$, we modify $f'$ into $f$ by setting $f(v)=[a_v+ \delta_v, b_v)$. Because $\sup\bigcup_{v\in V}f(v)=\sup\bigcup_{v\in V}f'(v)$, $f$ is an optimal extended strategy function for $T$.
\end{proof}

Given an extended strategy function $f$ for a node weighted tree $\inputTree{T}{\costFunc}$, one easily obtains a search strategy for the tree by first locating the vertex $q$ with the maximal interval assigned by $f$. In a trivial case, where $\vert V(T)\vert =1$, this completes the construction of $\strategy$. Otherwise, let us first note that there is only one such vertex $q$ in $V(T)$ with a maximal interval. Assume that there is a $q'\in V(T)$, $q\neq q'$, such that $b'=b$, where $f(q)=[a,b)$, $f(q')=[a',b')$. Then, because $f(q)\cap f(q')\neq \emptyset$, by the definition of $f$, on the simple path between $q$ and $q'$, a vertex $p$ would have to exist such that $f(p)>f(q)\cup f(q')$. Consequently, $q$ could not be maximal.
By the same argument, this maximal interval has an empty intersection with intervals assigned to any other vertex.
Let $\vert \neighbors{q}\vert =k$ and let $T_{i}$ be the $i$-th component of $T\setminus\{q\}$, $i\leq k$. Then, for arbitrary search target $\targetv$, $\strategy(T, \targetv)$ is a tuple $(q, \strategy(T_{1}, t), ..., \strategy(T_{k}, \targetv) )$.
Then, the construction of $\strategy$ continues recursively by using $f$ restricted to each component $T_{i}$.

On the other hand, given a search strategy $\strategy$, we can obtain an extended strategy function $f$ that encodes the same sequences of queries as those made by $\strategy$. Let $q$ be the vertex queried as first by $\strategy$. We set $f(q)=[0,\costFunc(q))$ in the trivial case, i.e. when $\vert V(T)\vert =1$. If $\vert V(T)\vert >1$, let $\strategy =(q, \strategy(T_{1}, t), ..., \strategy(T_{k}, \targetv))$, $\neighbors{q}=k$. We fix the values of $f$ for all vertices different than $q$ by recursively repeating the process for the components $T_{i}$. Note that each $v\in V(T)$ is referenced exactly once in $\strategy$. This is the case, because the components $T_{i}$ have disjoint vertex sets. After fixing $f$ in the components, we set $f(q)=[\sup\bigcup_{v\neq q}f(v),\sup\bigcup_{v\neq q}f(v)+\costFunc(q))$. We will now show that $f$ is an extended strategy function. Consider vertices $x\neq y$ in $V(T)$, such that $f(x)\cap f(y)\neq\emptyset$. Consider the sequences of queries $S_{x}$ and $S_{y}$, leading respectively to querying $x$ and $y$. Since both $S_{x}$ and $S_{y}$ are generated by the same strategy $\strategy$, they have a common prefix of length at least one of identical vertices. Let $v$ be the last vertex in the prefix. Since $x\neq y$, at least one vertex out of $x$ and $y$ was set after $v$ by our procedure before both $x$ and $y$. In other words, because $x\neq y$ only one of them might be $v$. Consequently we have that $f(x)\cup f(y)<f(v)$.

This also shows that the interval assigned to a vertex $v$ by $f$ might encode the time interval for querying $v$, where the `time' is understood as dictated by the precedence relation we use for intervals.

Given an extended strategy function $f$ for $T$, we say that a vertex $u$ is \emph{visible} from a vertex $v$ if and only if on the path that connects $u$ and $v$ there is no vertex $x$ such that $f(u)<f(x)$. The idea of visibility has been used e.g. in \cite{ZhouNN95} and also earlier under the name of `critical lists' in \cite{edge-rank-linear-LamY98}. The sequence of values (intervals) visible from $v$ in descending order is called a \emph{visibility sequence} (\emph{extended visibility sequence}) for $v$ and is denoted by $S(v)$.
Whenever we need to point out that $S(v)$ is obtained from a given $f$ in the above way, we say that $S(v)$ is \emph{derived} from $f$.

The concept of visibility is also used when $f$ is defined only on a subset of $V(T)$, i.e., on all descendants of some vertex $v$. In this context, which typically occurs during bottom-up processing of $T$, only those vertices are considered for which $f$ has been already defined.

We will decide on a strategy function $f(v)$ for each vertex $v$ by assigning an integer $j_v$ to $v$ in a bottom-up processing of the input tree. Each $j_v$ will determine $f(v)$ differently for different vertices. However, the $j_v$ will be computed via an appropriate use of the following operator.

\begin{definition}[\cite{Onak}] \label{def:vertex_extension}
    Let $f$ be a strategy function defined on all descendants of a $v\in V(T)$. Let $v_i$, $1\leq i\leq k$, be the children of $v$. The \emph{vertex extension} operator attributes $v$ with $f(v)$, based on the visibility sequences $S(v_i)$ in the following way. Let $M$ be the set of integers that belong to at least two distinct visibility sequences $S(v_i)$, while $m=\max(M\cup\{-1\})$. We set $f(v)$ as the lowest integer greater than $m$ that does not belong to any $S(v_i)$.
\end{definition}

The following lemma characterizes the vertex extension operator.

\begin{lemma}[\cite{Onak}] \label{lem:vertex-extension-is-minimizing}
	Given the visibility sequences assigned to the children of $v\in V(T)$, the \emph{vertex extension operator} assigns to $v$ a visibility sequence that is lexicographically minimal (the vertex extension operator is minimizing). Moreover, (lexicographically) increasing the visibility sequence of a child of $v$ does not decrease the visibility sequence calculated for $v$ (the vertex extension operator is monotone).
\end{lemma}

\section{Structured decision trees} \label{sec:structured_trees}

In Sections~\ref{sec:structured_trees}, \ref{sec:bottom-up-processing} we consider rooted trees $\inputTree{T}{\costFunc}$, in which while traversing any path from the root to a leaf, the weights of the visited vertices are non-increasing.
In particular, the root of $T$ is selected properly according to the definition of an up-monotonic cost function.
This section gives a foundation for the algorithm in Section~\ref{sec:algorithm-up-monotonic}.

We define a \emph{layer} of a tree $T$ as a subgraph $L$ of $T$ such that all vertices in $V(L)$ have the same cost, $\costFunc(u)=\costFunc(v)$ for any $u,v\in V(L)$. A connected component of $L$ is called a \emph{layer component}. We also denote by $\costFunc(L)$ the cost of querying a vertex that belongs to $L$, i.e., $\costFunc(L)=\costFunc(v)$ for any $v\in V(L)$.

The \emph{upper border} of a layer $L$ is the set of roots of its layer components. The subset of $V(L)$ consisting of vertices that have at least one child that belongs to a different layer is called the \emph{lower border} of $L$. We will also say that a layer component $L'$ is \emph{directly below} a layer component $L$ if and only if $\rootNode(L')$ has a parent in $L$. The \emph{top} layer component is the one that contains the root of $T$ and a \emph{bottom} layer component is any $L$ such that there is no component directly below $L$.
We note that we will apply the above terms to a rounded $\costFunc$.

\begin{definition} \label{def:structured-decision-tree}
	Let $T$ be a tree with a up-monotonic cost function and $\decTree{T}$ a decision tree for $T$. Consider a layer component $L$ of $T$. Let $v=\rootNode(L)$. We say that $\decTree{T}$ is \emph{structured with respect to} $L$ if for every vertex $u\in V(T_{v})$, it holds that $u\in V(\decTree{T}_{v})$.	
	We say that a decision tree $\decTree{T}$ for $T$ is \emph{structured} if and only if $\decTree{T}$ is structured with respect to all layer components of $T$.
\end{definition}
Informally speaking, in a structured $\decTree{T}$ we require that each vertex in the entire subtree $T_v$ is below $v$ in $\decTree{T}$, for each $v$ that is the root of a layer component.
This is one of the central ideas used in the algorithm for up-monotonic cost functions for the following reason.
While performing bottom-up processing of (an unweighted) $T$ in~\cite{Onak} (and also other works, e.g.~\cite{edge-rank-linear-LamY98}) each vertex $v$ encodes which steps are already used for queries of the vertices from $T_v$.
In our case, that is in a weighted $T$, we need to encode intervals and the technical problem with that is that they have different durations while moving upwards from one layer component to the next.
Thus, some intervals that are free to perform queries while moving to a parent of $v$ may be too short due to the weights in the next layer component.
To deal with that, we make $v=\rootNode(L)$ also the root of $\decTree{T}_v$.
In this way only one interval, i.e., the one assigned to $v$ needs to be taken into account while moving upwards.

Our method requires the decision tree to be structured only when processing the internal components of $T$. By not enforcing structuring with respect to the top layer component, a potential additive cost of a single query can be avoided. (This will be reflected in line~\ref{alg:line:optimization} of Algorithm \ref{alg:process-component}). However, this optimization does not change the worst-case cost of the strategy. In result, the decision tree generated by our algorithm is not structured in general. Although $\rootNode(\decTree{T})$, the first vertex queried by the strategy, always belongs to the top layer component of $T$, we may have that $\rootNode(\decTree{T}) \neq \rootNode(T)$.
In contrast, in a structured decision tree we always have a stronger property, namely $\rootNode(\decTree{T}) = \rootNode(T)$.

It turns out in our analysis that considering only structured decision trees introduces an overall multiplicative cost of $2$ to the performance of the algorithm:

\begin{lemma} \label{lem:structuring}
	Let $T=(V,E,\costFunc)$ be a tree with a up-monotonic and rounded cost function. There exists a structured decision tree $\decTree{T'}$ for $T$, such that $\COST{\decTree{T'}}\leq 2\:\OPT{T}$.
\end{lemma}
\begin{proof}
	Let $\decTree{T}$ be any decision tree for $T$. Since $\decTree{T}$ is selected arbitrarily, it is enough to prove that $\COST{\decTree{T'}}\leq 2\:\COST{\decTree{T}}$.
	
	In order to construct the structured decision tree, we traverse $\decTree{T}$ in a breadth-first fashion starting at the root of $\decTree{T}$. The $\decTree{T'}$ is constructed in the process.
	
	By Lemma~\ref{lem:component-in-dec-tree}, for an arbitrary layer component $L$ of $T$, there exists a vertex $u'\in V(L)$ such that all vertices from $L$ belong to $\decTree{T}_{u'}$. Among all vertices of $L$, $u'$ is visited first during the traversal.
	
	Consider an arbitrary node $u$ being accessed during the traversal. Let $L$ be the layer component such that $u\in L$. If $\decTree{T}$ is structured with respect to $L$ we continue with the next node. Otherwise (which for a specific $L$ will happen only once, when $u=u'$, as argued in the last paragraph of the proof), we modify $\decTree{T}$ by performing the following steps. 
	
	Let $v=\rootNode(L)$. We want to replace the subtree $\decTree{T}_{u}$ with a subtree consisting of the same nodes but rooted at $v$. Consider a subtree $T'$ of the input $T$, which can still contain the target when the search process is about to query $u$ according to $\decTree{T}$. E.g., $V(T')=V(\decTree{T}_{u})$. Let $v$ has $k$ neighbors in $T'$. We attach under $v$ the decision trees created with Lemma~\ref{lem:dec-tree-for-subtree} for the $k$ components obtained by removing $v$ from $T'$.
	Lemma~\ref{lem:dec-tree-for-subtree} shows that the cost of any of the $k$ decision trees is not greater than $\COST{\decTree{T}_{u}}$. Accounting for the cost of $v$, structuring $L$ increases $\COST{\decTree{T}}$ no more than $\costFunc(L)$. Consider an arbitrary path $P$ from root to a leaf in $\decTree{T}$. For every node $u\in P$ that triggers structuring (informally, $u$ is $u'$ for some $L$) an additional node from $L$ is inserted in $P$. Thus, in the worst case, the weighted length of arbitrary $P$ increases 2 times when transforming $\decTree{T}$ into structured $\decTree{T'}$.
	
	By starting the breadth-first traversal at the root of $\decTree{T}$, structuring with regards to one layer component does not make $\decTree{T}$ not structured with regards to any of the previously structured layer components.
\end{proof}

By combining Lemmas~\ref{lem:2rounded_dec+tree_cost} and~\ref{lem:structuring} one obtains Corollary~\ref{cor:exists}, which relates the cost of an optimal decision tree for an input $T$ with the cost of a structured decision tree obtained through the analysis of the rounding of $T$.

\begin{corollary} \label{cor:exists}
	For any tree $T$ with a up-monotonic cost function there exists a structured decision tree with rounded weights, whose cost is at most $4\cdot\OPT{T}$.
\end{corollary}

\section{Bottom-up tree processing} \label{sec:bottom-up-processing}

We extend the state-of-the-art ranking-based method for searching in trees with uniform costs of queries \cite{Onak,node-rank-linear-Schaffer89}, where the algorithm calculates a strategy function $f\colon V(T) \to \mathbb{N}$, given the input tree $T=(V,E)$. In our approach, to express the variable costs of queries we will use an interval-based extended strategy function.

Keeping in mind that the cost function is the time needed to perform a query, the values of an extended strategy function indicate the time periods in which the respective vertices will be queried.
Note that, e.g., due to the rounding, the time periods will be longer than the query durations and this is easily resolved either by introducing idle times or by following the actual query durations while performing a search based on a strategy produced by our algorithm.

In our approach, informally speaking, we will partition the rounding of the input tree into several subtrees (the layer components) and for each of them, we will use the algorithm from~\cite{Onak}.

A useful class of extended strategy functions is one that generates decision trees that are structured (see Definition~\ref{def:structured-decision-tree}).
\begin{definition} \label{def:structured-strategy-function}
	We say that an extended strategy function $f$, defined on the vertices of a tree $T$ with up-monotonic cost function, is \emph{structured} if for any layer component $L$, $f(u) > f(v)$ for each $v\in V(T_{u})$, where $u=\rootNode(L)$.
\end{definition}
It follows that a structured extended strategy function represents a structured decision tree. 

As mentioned earlier, we will process each layer component $L$ with the vertex extension operator but in order to do it correctly, we need to initialize appropriately the roots of the layer components below $L$.
For that, we define the following operators.
Let $f$ be an extended strategy function defined on the vertices of $T_{v}$, where $v$ is the root of a layer component $L'$. The \emph{structuring operator} assigns to $v$ the minimal interval so that $f$ is a structured extended strategy function on $T_{v}$. (We note that in our algorithm $v$ will have some interval assigned when the structuring operator is about to be applied to $v$. Hence, the application of this operator will assign the new required interval.) The \emph{cost scaling operator} aligns the interval attributed to $v$: if $f(v)=[a,b)$ and $L$ is the layer component directly above $L'$, the cost scaling operator assigns $v$ the interval  $[\costFunc(L)\lceil \frac{b}{\costFunc(L)} \rceil - \costFunc(L), \costFunc(L)\lceil \frac{b}{\costFunc(L)} \rceil)$.
We will say that vertex $v$ is \emph{aligned to} $\costFunc(L)$ after this modification (see the beginning of Section~\ref{sec:aligned-strategies} for a formal definition).
Informally speaking, this is done so that the children of the leaves in $L$ have intervals whose endpoints are multiples of $\costFunc(L)$ so that they become `compatible' with the allowed placements for intervals of the vertices from $L$.
Yet in other words, this allows us to translate the intervals of the children of the leaves in $L$ into integers that can be treated during bottom-up processing of $L$ in a uniform way.

\begin{observation} \label{obs:cost-of-cost-scaling}
	Let a layer component $L'$ be directly below a layer component $L$ in $T=(V,E,\costFunc)$. If $\rootNode(L')$ is assigned an interval $[a, b)$, whose endpoints are consecutive multiples of $\costFunc(L')$ (i.e. $a=k\costFunc(L')$ and $b=(k+1)\costFunc(L')$ for some integer $k$), then the cost scaling operator assigns to $\rootNode(L')$ an interval $[a', b')$ such that $b' \leq b + \costFunc(L) - \costFunc(L')$.
\end{observation}
\begin{proof}
	Consider a layer component $L$ of $T=(V,E,\costFunc)$ and a structured extended strategy function $f$ defined on the vertices below $L$. The cost scaling operator selects $b'$ as the smallest multiple of $\costFunc(L)$ that is greater or equal than $b$. Since $a'=b'-\costFunc(L)$, we have that $a'< b$. On the other hand, because $b-a=\costFunc(L')$, $a$ is the highest multiple of $\costFunc(L')$ that is less than $b$. Subsequently, because $w(L)$ is divisible by $w(L')$, we obtain that $a'\leq a$. What follows, $a'+\costFunc(L')\leq b$, and finally $b'\leq b + \costFunc(L)-\costFunc(L')$.
\end{proof}

\section{Aligned strategies} \label{sec:aligned-strategies}

In this section we give foundations for analyzing trees with down-monotonic cost functions, i.e. where the weights of the visited vertices are non-decreasing while traversing any path from the root to a leaf. In other words, we consider rooted trees $\inputTree{T}{\costFunc}$, where the root of $T$ is selected properly according to the definition of a down-monotonic cost function. Concepts provided in this section are used in Section~\ref{sec:algorithm-down-monotonic}, for the analysis of Algorithm~\ref{alg:down-monotonic}.

The first thing the Algorithm~\ref{alg:down-monotonic} does is to switch from the input tree to its rounding. Although the terms below apply to general weight functions, we should remember that they will be applied to the rounded ones in our analysis.

Given an extended strategy function $f$, we say that the vertex $v$ is \emph{aligned} if the endpoints of the interval $f(v)$ are multiples of $\costFunc(v)$. We also say that $f$ is \emph{aligned} if all vertices are aligned.

For an extended strategy function $f$, a vertex $u$ is \emph{screening} a vertex $v$ if and only if, on the path that connects $v$ and $u$, $u$ is the only vertex such that $f(v)<f(u)$.
In other words, for any path $P$ that starts with $v$, $u\in V(P)$ is screening $v$ if it is the closest vertex to $v$, such that $f(v)<f(u)$.
\begin{observation} \label{obs:screening-vertices}
    If $s,s'$ are distinct vertices screening some vertex, then $f(s)\cap f(s')=\emptyset$.
    \qed
\end{observation}

Furthermore, for an arbitrary vertex $v$, the \emph{screening neighborhood} of $v$ is the set $\visibilityN{v}\subseteq V(T)$ defined as follows.
We have that $u\in\visibilityN{v}$ if and only if either $u$ is screening $v$ or the path that connects $u$ and $v$ has no vertex screening $v$.
The subtree $T[\visibilityN{v}]$ is denoted by $\visibilityT{v}$.
For an example see Figure~\ref{fig:visibility}.
By definition, we include $v$ to be in $\visibilityN{v}$. Some of the properties of the screening neighborhood are as follows.

\begin{figure}
    \centering
    \includegraphics[scale=1.3]{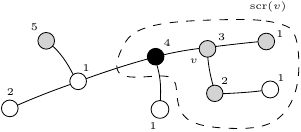}
    \caption{Example of a screening neighborhood. Labels at the vertices show the corresponding values of some strategy function. The dashed line encloses the screening neighborhood of the vertex $v$. The (single) vertex screening  $v$ is filled in black. All remaining vertices visible from $v$ are filled with light gray. Note that the vertex with label $5>f(v)$ is visible from $v$, although it does not belong to the screening neighborhood of $v$. All vertices that are not visible from $v$ are white.
    }
    \label{fig:visibility}
\end{figure}

\begin{observation} \label{obs:leafs}
    If $u\in \visibilityN{v}$ is a leaf of $\visibilityT{v}$, then we have that $u$ is screening $v$ or is a leaf in $T$ (or both).
    \qed
\end{observation}

As stated in Observation~\ref{obs:visN-and-visibility}, it may happen that a vertex is visible from $v$, but does not belong to $\visibilityN{v}$. Consider a vertex $u$ visible from $v$ with $f(v)<f(u)$, and a vertex $u'$ such that $u$ lies on the path between $v$ and $u'$ and $f(u)<f(u')$. Even if $u'$ was visible from $v$ (which then implies $f(v)<f(u')$), it would not belong to $\visibilityN{v}$.
In general, not every vertex visible from $v$ belongs to $\visibilityN{v}$. At the same time not every vertex in $\visibilityN{v}$ must be visible from $v$. 

\begin{observation} \label{obs:visN-and-visibility}
If a vertex $u$ is visible from $v$ and $f(u)<f(v)$, then $u\in\visibilityN{v}$.
\qed
\end{observation}
We will several times use an important property that for any $u\in \visibilityN{v}\setminus\{v\}$, either $f(u)<f(v)$ or $f(v)<f(u)$.
This fact basically allows to simplify arguments by reducing the number of possible subcases in the following lemmas.
\begin{observation} \label{obs:vertices_in_visibilityN}
     For each $u\in \visibilityN{v}$, $u\neq v$, it holds $f(v)\cap f(u)=\emptyset$. 
\end{observation}
\begin{proof}
    By contradiction, assume a $u\in \visibilityN{v}$ exists such that $f(u)\cap f(v) \neq \emptyset$. Since $f$ is an extended strategy function, there must be an $x$ on the path between $u$ and $v$, such that $f(x)>f(v)$.
    Among all such vertices $x$ consider one with the maximum value of $f(x)$.
    This choice implies that $f(x)>f(v)$ and there is no vertex $y$ on the path between $v$ and $x$ such that $f(y)>f(x)\cup f(v)$.
    This gives that $x$ is screening $v$ contradicting $u\in\visibilityN{v}$.
\end{proof}

\begin{lemma} \label{lem:independent_strategies}
    Let $f$ be an extended strategy function for $T=(V,E,\costFunc)$. Consider an arbitrary $v\in V(T)$. Let $s_{min}\in \visibilityN{v}$ be the vertex screening $v$ with the minimal value of $f$, among all vertices screening $v$.
    Consider an arbitrary function $f'\colon V(T) \to \{[a,b)\colon0\leq a<b\}$, $\rvert f'(v)\rvert \geq\costFunc(v)$, which satisfies the following conditions:
    \begin{enumerate}
        \item $f'(u)=f(u)$, for any $u\not\in\visibilityN{v}$,
        \item $f'(u)=f(u)$, if $u\in\visibilityN{v}$ is screening $v$,
        \item $f'(u) < f(s_{min})$, for any $u\in\visibilityN{v}$ that is not screening $v$. (This includes $v$.)
    \end{enumerate}
    We have the following: if $f'$ is an extended strategy function for $\visibilityT{v}$, then $f'$ is also an extended strategy function for the entire $T$.
\end{lemma}
\begin{proof}
The proof is by case analysis dictated by the definition of visibility function.
Consider two vertices $x,y\in V(T)$ such that $x\neq y$ and $f'(x)\cap f'(y)\neq \emptyset$. For each possible choice of $x$ and $y$, we show that on the path connecting $x$ and $y$, there is a vertex $z$ such that $f'(x)\cup f'(y)<f'(z)$.
Note that based on Observation~\ref{obs:screening-vertices}, $s_{min}$ is unique in $\visibilityN{v}$.
\begin{itemize}
    \item Let $x,y\in \visibilityN{v}$.
        By assumption, $f'$ is an extended strategy function for $\visibilityT{v}$. Thus, on the path between $x$ and $y$ there exists a required vertex $z$.
    \item Let $x,y\not\in\visibilityN{v}$.
        For all vertices outside of $\visibilityT{v}$, $f'$ is identical to $f$. Because $f$ is an extended strategy function, there exists $z\in V(T)$, such that $f(x)\cup f(y)<f(z)$.
        \begin{enumerate}
            \item If $z$ is outside of $\visibilityT{v}$, then $f'(z)=f(z)$ and consequently $f'(x)\cup f'(y)<f'(z)$.
            \item If $z\in \visibilityN{v}$ and $z$ is screening $v$, then also in this case $f'(z)=f(z)$ and consequently $f'(x)\cup f'(y)<f'(z)$.
            \item If $z\in\visibilityN{v}$ and $z$ is not screening $v$, let $P$ be the path connecting $x$ and $y$. (See Figure~\ref{fig:lemma11} for an illustration.)
            \begin{figure} \label {fig::visibility}
                \centering
                \includegraphics[scale=0.8]{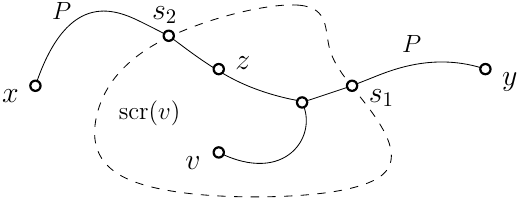}
                \caption{An illustration for the case when $x,y\not\in\visibilityN{v}$ in Lemma~\ref{lem:independent_strategies}.}
                \label{fig:lemma11}
            \end{figure}
            Because the endpoints $x,y\not\in\visibilityN{v}$ and $z\in \visibilityN{v}$ is not screening $v$, it follows from the definition of $\visibilityN{v}$ that $P$ contains vertices $s_1$ and $s_2$ screening $v$, such that the path between $s_1$ and $s_2$ contains $z$. Let $s=\textup{argmax}_{s\in\{s_1,s_2\}}f(s)$. Because $z$ is on a path either between $v$ and $s_1$ or $v$ and $s_2$, based on the definition of a screening vertex, we have that $f(z)<f(s)$.                        
            Subsequently, we have that $f'(x)=f(x)<f(z)<f(s)=f'(s)$. By similar arguments, $f'(y)<f'(s)$. Hence, on the path between $x$ and $y$, $s$ is the required vertex separating $x$ and $y$.
        \end{enumerate}
    \item Let $x\in \visibilityN{v}$ and $y\notin \visibilityN{v}$. Let $P$ be the path connecting $x$ and $y$. Based on the definition of $\visibilityN{v}$, $P$ contains at least one vertex screening $v$ (otherwise, for  $x\in \visibilityN{v}$, $y$ would have to belong to $\visibilityN{v}$ too). Let $s\in V(P)$ be the vertex screening $v$ that is closest to $y$. We note that $P$ contains two vertices screening $v$ only when $x$ is screening $v$.
    \begin{enumerate}
        \item Let $x$ be screening $v$. We have $f'(x)=f(x)$. Because $f$ is an extended strategy function for $T$, there is $z\in V(P)$, such that $f(x)\cup f(y)<f(z)$. We consider two subcases.
            \subitem Let $x=s$. Recall that $s\in V(P)$ is the vertex screening $v$ that is closest to $y$. Hence, $x$ is the only vertex in $P$ that belongs to $\visibilityN{v}$. Subsequently, $z\notin\visibilityN{v}$ and we have that $f'(z)=f(z)$. What follows, $f'(x)=f(x)<f(z)=f'(z)$.
            \subitem Let $x\neq s$. If $z\notin\visibilityN{v}$ or $z$ is screening $v$ (i.e., $z$=$s$), then $f'(z)=f(z)$ and $z$ is the required vertex separating $x$ and $y$. If $z\in\visibilityN{v}$ and $z$ is not screening $v$, we have that $f(z)<f(s)$. What follows, $f'(x)=f(x)<f(z)<f(s)=f'(s)$.
        \item If $x$ is not screening $v$, then by definition of $f'$ and $s_{min}$, $f'(x)<f'(s_{min})\leq f'(s)$. What follows, $s$ is the required vertex separating $x$ and $y$.
    \end{enumerate}
\end{itemize}
\end{proof}

The technical Lemma~\ref{lem:technical-property} shows a property of the screening neighborhood of specific unaligned vertices. This property is later used in the proof of Lemma~\ref{lem:aligned_strategy}.
\begin{lemma} \label{lem:technical-property}
    Consider an extended strategy function $f$ for $T=(V,E,\costFunc)$. Let $v\in V(T)$ be an unaligned vertex, such that $f(v)$ is minimal. Denote $f(v)=[x_v, y_v)$, $|\costFunc(v)|=2^{k}$. Take $a\in \mathbb{N}$ such that $a2^{k}<x_v<(a+1)2^{k}$. For each $u\in \visibilityN{v}$ such that $f(u)\cap [a2^{k}, x_v)\neq \emptyset$, it holds $f(u)\subset [a2^{k}, x_v)$.
\end{lemma}
\begin{proof}
    We know that $a$ exists because $v$ is unaligned.
    By Observation~\ref{obs:optimal-strategy-function}, $|f(v)|=2^k$.
    It follows from Observation~\ref{obs:vertices_in_visibilityN} that $f(u)<f(v)$ for any $u$ as in the statement of the lemma, and consequently, that $u$ is aligned.
    
    Let $f(u)=[x_u,y_u)$. By assumption and by $f(u)<f(v)$, we know that $a2^k< y_u < x_v$.
    Since $u$ is aligned, we have that $x_u=b2^{k'}$ and $y_u=(b+1)2^{k'}$, for some integer $k'$.
        Note that $a2^k< y_u < x_v<(a+1)2^k$ implies $k'<k$.
        Thus, $y_u\geq a2^k+2^{k'}$, because $u$ is aligned.
        This gives $x_u=y_u-2^{k'}\geq a2^k$, as required in the lemma.
\end{proof}

We now prove the main result of this section:
\begin{lemma} \label{lem:aligned_strategy}
    For each instance with down-monotonic and rounded cost function, there exists an optimal extended search strategy function that is aligned.
\end{lemma}

\begin{proof}
    In the proof, we consider an optimal extended strategy function $f$, such that the number of unaligned vertices is minimal.
    Based on $f$, we construct another optimal extended strategy function according to which, there are fewer unaligned vertices. This shows, by contradiction, that there must exist an optimal extended strategy function such that all vertices are aligned.

\medskip
Let $v\in V(T)$ be an unaligned vertex, such that the interval $f(v)$ is minimal. Let $\costFunc(v)=2^{k}$, $f(v)=[x_v,y_v)$. Take $a\in \mathbb{N}$, such that $a2^{k}<x_v<(a+1)2^{k}$ and $(a+1)2^{k}<y_v<(a+2)2^{k}$.
\medskip

Based on the strategy function $f$, we will now construct a $f'$ with fewer unaligned vertices.

\begin{enumerate}
    \item For all $u\notin \visibilityN{v}$, we set $f'(u)=f(u)$.
    \item If $s\in \visibilityN{v}$ is screening $v$, we set $f'(s)=f(s)$.
    \item For all other vertices in $\visibilityN{v}$:
    \subitem For the $v$ itself, we set $f'(v)=[a2^{k}, (a+1)2^{k})$.
    \subitem For any vertex $u\in \visibilityN{v}$ such that $f(u)\cap [a2^{k}, x_v)\neq \emptyset$, we set $f'(u)=[x_u+2^k,y_u+2^k)$, where $f(u)=[x_u,y_u)$.
    \subitem For all other vertices $u\in\visibilityN{v}$ we set $f'(u)=f(u)$.
\end{enumerate}

    For brevity, we will refer to the vertices $u\in V(T)$ for which $f'(u)\neq f(u)$ as \emph{updated} vertices. We will refer to all other vertices as \emph{still} vertices.
    A useful fact about the updated vertices is the following:
    \begin{enumerate}[label={\normalfont{(*)}},leftmargin=*]
    \item\label{fact:updated-vertices} We have $f'(u)<y_v$, for every updated vertex $u\in V(T)$.
    \end{enumerate}
    By definition of $f'$, updated vertices are those whose value of $f$ intersects with $[a2^{k}, x_v)$ or the $v$ itself. The fact \ref{fact:updated-vertices} holds for $v$. Consider $u\in V(T)$, such that $f(u)\cap [a2^{k}, x_v)\neq\emptyset$. Denote $f(u)=[x_u, y_u)$, $f'(u)=[x_u', y_u')$. By Lemma~\ref{lem:technical-property}, we have that $y_u\leq x_v$. Subsequently, $y_u'\leq x_v+2^k=y_v$. This completes the proof of \ref{fact:updated-vertices}.
    
    Let us note that $v$ is aligned according to $f'$ and all vertices which were aligned according to $f$ are also aligned according to $f'$.
    To see the latter, consider an updated vertex $u\neq v$.
    If $u$ is aligned according to $f$, then by the definition of $f'$, $u\in\visibilityN{v}$ and $f(u)\cap[a2^k,x_v)$.
    But then, by Lemma~\ref{lem:technical-property}, $f(u)\subseteq[a2^k,x_v)$ which in particular means that $\costFunc(u)\leq 2^k$ and hence $u$ is indeed aligned according to $f'$.
    We also observe that $\sup{f'(v)}<\sup{f(v)}$, and for each $u$, such that $f(u)\cap [a2^{k}, x_v)\neq \emptyset$, we have $\sup{f'(u)}\leq\sup{f(v)}$.
    Hence, $f'$ is optimal due to the optimality of $f$.
    The above in particular implies that it remains to argue that $f'$ is a valid extended strategy function.

    We will now show that $f'$ is an extended strategy function for $T$. Based on Lemma~\ref{lem:independent_strategies}, it suffices to show that $f'$ is an extended strategy function for the subtree $\visibilityT{v}$. Consider vertices $\tempU,\tempV \in\visibilityN{v}$, such that $\tempU\neq\tempV$ and $f'(\tempU)\cap f'(\tempV)\neq \emptyset$. We analyze all possible choices of the pair $\tempU$, $\tempV$. Motivated by Observation~\ref{obs:vertices_in_visibilityN}, we categorize the choices for selecting $\tempU$ depending on how $f'(\tempU)$ relates to the interval $f(v)=[x_v,y_v)$.

\medskip\noindent
    Case 1: Let $\tempU$ be such that $y_v<f'(\tempU)$.
    We argue that
    \begin{equation} \label{eq:tempV}
    f(v) < f(\tempV)
    \end{equation}
    Suppose otherwise, which in view of Observation~\ref{obs:vertices_in_visibilityN} implies $f(\tempV) < f(v)$.
    Since $\tempU$ is still by \ref{fact:updated-vertices}, $f(\tempU)=f'(\tempU)$. Then, $y_v<f(\tempU)$ can be written as $f(v)<f(\tempU)$. But then $f(\tempV)<f(v)$ gives $f(\tempV)<f(\tempU)$, which contradicts $f(\tempU)\cap f(\tempV)\neq\emptyset$ and implies~\eqref{eq:tempV}. The latter and \ref{fact:updated-vertices} prove that $\tempV$ is still, thus $y_v<f'(\tempV)$.

    Because $f$ is an extended strategy function, there is a vertex $\tempX\in \visibilityN{v}$ on the path between $\tempU$ and $\tempV$, such that $f(\tempX)>f(\tempU)\cup f(\tempV)$. Because $f(\tempX)>f(\tempU)>f(v)$, we have by \ref{fact:updated-vertices} that $\tempX$ is still. Subsequently, $f'(\tempU)=f(\tempU)<f(\tempX)=f'(\tempX)$ and since $\tempV$ is still $f'(\tempV)=f(\tempV)<f(\tempX)=f'(\tempX)$, thus $f'(\tempX)>f'(\tempU)\cup f'(\tempV)$.
    
\medskip\noindent
    Case 2: Let $\tempU$ be such that $f'(\tempU)<a2^{k}$. By repeating the steps in Case 1, i.e. when $y_v<f'(\tempU)$, and using Lemma~\ref{lem:technical-property}, we obtain that $f'(\tempV)<a2^k$. Because $f$ is an extended strategy function, there is a vertex $\tempX\in \visibilityN{v}$ on the path between $\tempU$ and $\tempV$, such that $f(\tempX)>f(\tempU)\cup f(\tempV)$.
    If $\tempX$ is an updated vertex, we have that $f'(\tempX)>a2^{k}$, hence $f'(\tempU)<f'(\tempX)$. Otherwise, we have $f'(\tempU)=f(\tempU)$, and the claim follows from $f(\tempU)=f'(\tempU)$ and $f(\tempV)=f'(\tempV)$.

\medskip\noindent
        Case 3: Let $\tempU$ be such that $f'(\tempU)\cap [(a+1)2^{k}, y_v)\neq \emptyset$. Note that $\tempU\neq v$ because by definition of $f'$, $f'(v)=[a2^k,(a+1)2^k)$. Because $f'(\tempU)\cap f(v)\neq\emptyset$, it follows from Observation~\ref{obs:vertices_in_visibilityN} that $\tempU$ is updated. Subsequently, we have that $f(\tempU)\cap [a2^{k}, x_v)\neq \emptyset$. Based on Lemma~\ref{lem:technical-property}, it holds that $f(\tempU)\subset [a2^{k}, x_v)$ and subsequently that $f'(\tempU)\subset [(a+1)2^{k}, (a+2)2^k)$. By definition, $f'(\tempV)\cap f'(\tempU)\neq\emptyset$. Thus, $f'(\tempV)\cap [(a+1)2^{k}, y_v) \neq\emptyset$. Following the same steps as for $\tempU$, we obtain that $\tempV$ is updated and that $f'(\tempV)\subset [(a+1)2^{k}, (a+2)2^k)$.
        
        Because $f$ is an extended strategy function, there is $\tempX\in \visibilityN{v}$ on the path between $\tempU$ and $\tempV$, such that $f(\tempX)>f(\tempU)\cup f(\tempV)$.
        Because both $f(\tempU)$ and $f(\tempV)$ are subsets of $[a2^{k}, x_v)$, we have that $\costFunc(\tempU)< \costFunc(v)$ and $\costFunc(\tempV)<\costFunc(v)$. Because $\costFunc$ is down-monotonic, we have that neither $\tempU$ nor $\tempV$ belong to $T_v$. What follows, $\tempX\neq v$. If $\tempX$ is an updated vertex, then we have $f'(\tempX)>f'(\tempU)$. Otherwise, because $z\neq v$, $f'(\tempX)=f(\tempX)>y_v$, and we have that $f'(\tempU)<y_v<f'(\tempX)$. Following the same steps for $\tempV$, we obtain $f'(\tempX)>f'(\tempU)\cup f'(\tempV)$.

\medskip\noindent
        Case 4: Let $\tempU$ be such that $f'(\tempU)\cap[a2^k,(a+1)2^k)\neq\emptyset$. We will start by showing that this condition holds only when $u^*=v$.

        Consider the following ways we may choose $\tempU\in\visibilityN{v}$, with respect to the interval $f(\tempU)$:
        \begin{enumerate}
            \item $f(\tempU)<a2^k$.           
            By definition of $f'$, vertex $\tempU$ is still, thus $f'(\tempU)<a2^k$. What follows, such $u$ does not exist.
            
            \item $f(\tempU)\cap[a2^k,x_v)\neq\emptyset$.
            By Lemma~\ref{lem:technical-property}, we have that $f(\tempU)\subset [a2^k,x_v)$. By definition of $f'$, we have that $(a+1)2^k\leq f'(\tempU)$.

            \item $f(\tempU)\cap [x_v,y_v)\neq\emptyset$.
            Based on Observation~\ref{obs:vertices_in_visibilityN}, we have $\tempU= v$.
            
            \item $y_v\leq f(\tempU)$.            
            By definition of $f'$, $\tempU$ is still, thus $(a+1)2^k<y_v\leq f(\tempU)=f'(\tempU)$.
        \end{enumerate}

        Subsequently, because $\tempV\neq \tempU$, we have that $f'(\tempV)\cap[a2^k,(a+1)2^k)=\emptyset$. This shows that it is not possible to find a $\tempV\in \visibilityN{v}$, such that because $f'(\tempU)\cap f'(\tempV)\neq \emptyset$.
\end{proof}

\section{The algorithm for up-monotonic cost functions}
\label{sec:algorithm-up-monotonic}

In this section we propose an algorithm for trees with up-monotonic cost functions (cf. Algorithm~\ref{alg:up-monotonic}). The algorithm starts with a pre-processing pass, which transforms the input tree to its rounding and selects a root arbitrarily from the set of vertices with the highest cost o query (cf. definition of an up-monotonic cost function, Section~\ref{sec:our-results}). An extended strategy function $f$ is then calculated during a bottom-up traversal over the layer components of the tree by applying Algorithm \ref{alg:process-component}. What follows, the extended strategy function for the vertices of a layer component $L$ is calculated only after it is already defined for all vertices below the component.
Moreover, since we use a structured $f$, only its values assigned to the roots of the layer components below $L$ are important when calculating $f$ for the vertices of $L$.

\begin{figure}[htb]
\begin{center}
\begin{minipage}{.9\linewidth}
\begin{algorithm}[H]
\SetAlgoRefName{Up-Monotonic}
    \caption{(input: tree $T$ with up-monotonic cost function; output: an extended strategy function $f$ calculated for $T$)}
    \label{alg:up-monotonic}
    $T\gets$ a rooted rounding of $T$\;
    \For{each layer component $L$ in bottom-up fashion}{
		Apply Algorithm~\ref{alg:process-component} to $L$\;
    }
    \textbf{return} $f$
\end{algorithm}
\end{minipage}
\end{center}
\end{figure}

Given a layer component $L$, Algorithm \ref{alg:process-component} iterates over the vertices of $L$ in a depth-first, postorder fashion.
Suppose that a structured extended strategy function $f$ is defined on all vertices below $L$ and that $f$ is aligned to $\costFunc(L)$ at all roots of layer components directly below $L$. To calculate $f$ for the vertices of $L$, we first calculate the strategy function $f'$ for each vertex of $L$ and then obtain $f$ from $f'$ using the formula: $f(u)=[f'(u)\costFunc(L), (f'(u)+1)\costFunc(L))$, where $u\in V(L)$. (Recall that this is a conversion that takes us from the integer-valued strategy function $f'$ to the interval-valued extended strategy function.) To provide an applicable input for the vertex extension operator at the leaves of $L$, for each child $u\in V(L')$ of such leaf, we calculate the integer $f'(u)=\frac{b}{\costFunc(L')}$ that is derived from the extended strategy function at $u$, i.e. from $f(u)=[a,b)$. Note that since our decision trees are structured, the visibility sequence corresponding to each root of a layer component below $L$ (which includes $u$) consists of only one value, the one that equals the strategy function at the root. What follows, deriving the value of $f'$ only for the roots of the layer components directly below $L$ is sufficient to create a valid input for the vertex extension operator. In other words, we make these preparations to use the operator and the corresponding method from \cite{Onak}. Once each vertex $v\in V(L)$ obtains the corresponding extended strategy function $f(v)$, we apply for the root of $L$ the structuring operator followed by the cost scaling operator.
This will close the entire `cycle' of processing one layer component.

\begin{figure}[htb]
\begin{center}
\begin{minipage}{.9\linewidth}
\begin{algorithm}[H]
\SetAlgoRefName{Process-Component}
    \caption{(input: layer component $L$). Let $f$ be an extended strategy function defined for all descendants of vertices in $L$ that do not belong to $L$.}
    \label{alg:process-component}
    
    \For{each root $v$ of a layer component directly below $L$}
    {
        $f'(v)\gets \frac{b}{\costFunc(L)}$, where $b$ is the right endpoint of the interval $f(v)$\;
    }
    \For{each $u\in V(L)$ in a postorder fashion}
    {
        Obtain $f'(u)$ by applying the vertex extension operator to $v$\;
        $f(u)\gets$ the interval derived from $f'(u)$, see Section~\ref{sec:bottom-up-processing}\;
    }
    \If{the root $v$ of $L$ is not the root of $T$}{
        \label{alg:line:optimization}
        Update $f(v)$ by applying the structuring operator to $v$\;
        Update $f(v)$ by applying the cost scaling operator to $v$\;
    }
\end{algorithm}
\end{minipage}
\end{center}
\end{figure}

\subsection{Analysis} \label{sec:up-monotonic-analysis}

Lemma~\ref{lem:optimal-strategy} characterizes the interval of the extended strategy function that Algorithm \ref{alg:process-component} assigns to the root of a layer component.

\begin{lemma} \label{lem:optimal-strategy}
	Consider a layer component $L$ of $T=(V,E,\costFunc)$ and a structured extended strategy function $f$ defined on the vertices below $L$. Let the roots of the layer components directly below $L$ be aligned to $\costFunc(L)$.
	Suppose that the values of $f$ over $V(L)$ are obtained by a call to procedure $\textup{ProcessComponent(L,f)}$.
	We have that $f$ is a structured extended strategy function on $V(T_{\rootNode(L)})$, and the lowest possible interval $I$, such that $I>f(u)$, for each $u$ below $\rootNode(L)$, is assigned to $\rootNode(L)$.
\end{lemma}

\begin{proof}
	Due to the alignment of $f$ at the roots of the components directly below $L$, the strategy function $f'$ derived from $f$ for these roots is consistent with the query costs in $L$.
	That is, the construction of $f'$ done at the beginning of the procedure gives an integer-valued strategy function.
	Since $f$ is structured, among all vertices below $L$, the extension operator correctly needs to consider only the values of $f'$ at the roots of components directly below $L$.
	
	By Lemma~\ref{lem:vertex-extension-is-minimizing}, the visibility sequence calculated for the root of $L$ is (lexicographically) minimized among all possible valid assignments of $f'$.
	Hence in particular, the value of $\max\{f'(v) \;|\; v\in V(T_{\rootNode(L)})\}$ is minimized.
	It follows that if there exists such an optimal assignment of $f'$ that is structured, then this $f'$ is calculated by the procedure is structured and consequently the $f$ obtained from $f'$ is minimal and structured.
	If there exists no optimal assignments of $f'$ over $V(L)$ that is structured, the structuring operator modifies $f$ and assigns to $\rootNode(L)$ the lowest interval that is greater than $f(u)$ for any $u$ below $\rootNode(L)$.
\end{proof}

It follows from Lemma~\ref{lem:optimal-strategy} and Observation~\ref{obs:cost-of-cost-scaling} that given an extended strategy function $f$ fixed below a layer component $L$, Algorithm~\ref{alg:up-monotonic} extends $f$ to the vertices of $L$ in such a way that the interval assigned to $\rootNode(L)$ is not higher than the optimal interval (assuming that $f$ is fixed below $L$) incremented by an additive factor $\costFunc(L')-\costFunc(L)$, where $L'$ is a layer component directly above $L$. If $L$ is the top component, the structuring transform is not applied and the additional cost is not incurred.

We now formulate a technical Lemma~\ref{lem:generated_strategy_function} that helps analyze the strategy generated by Algorithm \ref{alg:up-monotonic}.

\begin{lemma} \label{lem:generated_strategy_function}
	Let $T=(V,E,\costFunc)$ be a tree with a up-monotonic and rounded cost function. We denote by $f_{opt}$ an optimal structured extended strategy function on $V(T)$. Let $L$ be a layer component of $T$ with $k$ layer components $L_{j}$, $1\leq j\leq k$, directly below $L$. Let $v_{j}=\rootNode(L_{j}), 1\leq j\leq k$.
	We define $f'$ as a function defined on $V(T)$ such that for all vertices below any $v_{j}$, $f'$ is equal to $f_{opt}$, while for the vertices $v_{j}$ and above, $f'$ is equal to $f_{opt}+\costFunc(L)$.
	It holds that $f'$ is a structured extended strategy function on $V(T)$.
\end{lemma}

\begin{proof}
	
	We first show that $f'$ is an extended strategy function. We need to show that for any pair $x,y \in V(T)$, $x\neq y$, if $f'(x)=f'(y)$, then there is a vertex $z\in V(T)$ separating $x,y$ such that $f'(z)>f'(x)$. (See Definition~\ref{def:extended_strategy_function}.)
	
	We analyze the following cases with respect to the locations of the vertices $x$ and $y$. In all cases we assume that $x\neq y$, $f'(x)=f'(y)$, and $z$ is the vertex separating $x$ and $y$ according to $f'$.
	\begin{itemize}
		\item $x\in V(L_i)$, $y\in V(L_j)$ and $i=j$.
		
		Let $v=\rootNode(L_i)$. Because $f_{opt}$ is structured, $f'$ is also structured within the subtree $T_v$, since it assigns $v$ a higher interval than $f_{opt}$ and $f'$ is otherwise identical to $f_{opt}$. What follows, because $f'$ is structured and we assumed $f'(x)=f'(y)$, we have that $x\neq v$ and $y\neq v$.
		
		Consider the $f_{opt}$ as defined on all vertices of $L_i$. Because $x,y\in V(L_i)$, we have $z\in V(L_i)$. If $z\neq v$, then $f'(z)=f_{opt}(z)>f_{opt}(x)=f'(x)$, where the last equation holds because $x\neq v$. If $z=v$, we also have $f'(z)>f'(x)$, because $f'$ is structured in $T_v$.

		\item $x\in V(L_i)$, $y\in V(L_j)$ and $i\neq j$.
		
		Because $f'$ is structured in $L_i$, we have $z=\rootNode(L_i)$.
		
		\item Both $x$ and $y$ belong to $V(L)$.
		
		Because $f_{opt}$ is an extended strategy function in $L$, there is a $z\in V(L)$ such that $f_{opt}(z)>f_{opt}(x)$. Then, by adding $\costFunc(L)$ to both sides of the inequality one obtains $f'(z)>f'(x)$.
		
		\item $x\in V(L_i)$ for some $i$ and $y\in V(L)$.

		If $x\neq \rootNode(L_i)$ then $z=\rootNode(L_i)$, because $f'$ is structured in $L_i$.

		Otherwise, both $x$ and $y$ are in the subtree $T[V(L)\cup\{x\}]$. Recall that in this subtree $f'$ is derived from $f_{opt}$ by adding a constant offset, which gives an extended strategy function.
	\end{itemize}

To show that $f'$ is structured in the remaining cases, we observe that for all vertices above $L$, $f'$ is derived from the structured $f_{opt}$ by shifting the assigned intervals by a positive offset, which does not change the structural property of the layer components' roots.
\end{proof}

\begin{lemma} \label{lem:structured_tree_approximation}
    The cost of the decision tree generated by Algorithm~\ref{alg:up-monotonic} is at most 2 times greater than the cost of an optimal structured decision tree.
\end{lemma}

\begin{proof}
	The proof is by induction, bottom-up over the layer components of the input tree. Take a tree $T=(V,E,\costFunc)$ with an up-monotonic and rounded cost function.
	We denote by $c(i), 1\leq i\leq l$, the query cost to a vertex in the $i-th$ layer, where the layers are ordered according to their cost of query and $l$ is the number of layers in $T$. Let $f_{opt}$ be an optimal (structured) extended strategy function and $f_{alg}$ be the extended strategy function generated by Algorithm~\ref{alg:up-monotonic}. We denote by $z(v)$ the cost of querying the parent of $v$ or the cost of $v$ itself if there is no vertex above $v$, i.e. when $v$ is the root of $T$.
	
	We want to prove the following induction claim. For any layer component $L$ it holds $f_{alg}(v) \leq f_{opt}(v) + z(v)$, where $v=\rootNode(L)$.
	
	The cost of the decision tree generated by $f_{alg}$ is equal to the supremum of $f_{alg}(\rootNode(T))$. Since the cost of a single query to the top layer is not more than the cost of an optimal structured decision tree, the lemma follows from the induction claim applied to the top layer component.
	
	Let $L$ be a bottom layer component. Based on Lemma~\ref{lem:optimal-strategy} we know that by applying the vertex extension and structuring operators, Algorithm~\ref{alg:up-monotonic} generates an optimal structured assignment of the extended strategy function to the vertices of $L$. In a case where $L$ is also the top component (which is the trivial case of $T$ being a single layer component) we have $f_{alg}$ is optimal and the induction claim holds. Otherwise, the cost scaling operator extends the interval assigned to $\rootNode(L)$, increasing the cost of $L$ by (at most) an additive factor of $c(2)-c(1)$ over the optimal cost. Since the added factor is less than $z(\rootNode(L))$, the induction claim holds for the bottom layer component $L$.

	Let $L$ be a layer component from the $i$-th layer, with $k$ layer components $L_{j}$, $1\leq j\leq k$, directly below $L$. We denote the roots of the layer components $L_j$ by $v_j$. Consider a strategy function $f'$ which is defined as in Lemma~\ref{lem:generated_strategy_function}, that is for all vertices below any $v_{j}$, $f'$ is equal to $f_{opt}$, while for vertices $v_{j}$ and all vertices above, $f'$ is equal to $f_{opt}+\costFunc(L)$. Lemma~\ref{lem:generated_strategy_function} implies that $f'$ is a structured extended strategy function on $V(T)$.
	
	From the induction hypothesis we have that for any $v_j$, $f_{alg}(v_{j})\leq f_{opt}(v_{j}) + z(v_{j})=f_{opt}(v_j)+\costFunc(L)$. According to Lemma~\ref{lem:optimal-strategy} and Observation~\ref{obs:cost-of-cost-scaling}, Algorithm~\ref{alg:up-monotonic} assigns intervals to the vertices of $L$ such that the interval assigned to the root of $L$ is the lowest possible (when $L$ is the top component), or the lowest possible interval incremented by a positive offset of $\costFunc(L')-\costFunc(L)$, where $L'$ is the layer component directly above $L$. If we sum up the additive increments due to cost scaling in all layer components below $L$, we obtain $\sum_{i<l}c(i)<c(l)=z(v)$ (recall that $c(i)$'s are powers of $2$).
	Thus, $f_{alg}(v)\leq f_{opt}(v) + z(v)$.
\end{proof}

	By combining
	Lemmas \ref{lem:2rounded_dec+tree_cost},~\ref{lem:structuring}, and~\ref{lem:structured_tree_approximation}
	we obtain that Algorithm~\ref{alg:up-monotonic} generates a solution for the binary search problem with a cost at most $8$ times greater than that of an optimal solution. We note that Algorithm~\ref{alg:up-monotonic} maintains the main structure of the linear time algorithm from \cite{Onak} (a single bottom-up pass over vertices of the tree.) Our method extends the state-of-the-art algorithm by adding a fixed number of $\bigo{1}$ steps, computed when visiting vertices during the bottom-up traversal. See Algorithm \ref{alg:process-component}. What follows, Algorithm~\ref{alg:up-monotonic} also runs in linear time.
	This proves Theorem~\ref{thm:up-monotonic}.

\section{The algorithm for down-monotonic cost functions} \label{sec:algorithm-down-monotonic}
In this section, we describe and analyze the algorithm providing Theorem~\ref{thm:down-monotonic}.
To construct an interval-based extended strategy function $f$, our procedure adapts the greedy linear-time algorithm for searching in trees with uniform query costs~\cite{Onak}.
Let $T=(V,E,\costFunc)$ be a rooted tree, such that $\costFunc$ is rounded and down-monotonic. Recall that the root must be one of the vertices of minimal weight.

In our algorithm that takes $T$ as input, the intervals that $f$ assigns to the vertices are calculated in a bottom-up fashion.
Consider a vertex $v\in V(T)$ with $k>0$ children $v_{1},\ldots,v_{k}$ so that $f$ is defined on all descendants of $v$.
For each $v_i$, we denote the extended visibility sequence derived from $f$ restricted to $T_{v_i}$ by $S(v_i)$.
Let $\cI$ be the set of all intervals $I$ that belong to any $S(v_i)$, such that there exists an interval $I'\in S(v_j)$, $i\neq j$, where $I\cap I'\neq\emptyset$. If $\cI\neq\emptyset$, let $I_M$ be the maximal interval in $\cI$. If $\cI=\emptyset$, then let $I_M=[-1,0)$. The minimal interval $I$ with endpoints being multiples of $\costFunc(v)$, such that $I>I_M$, $|I|=\costFunc(v)$ and $I\cap I'=\emptyset$ for each $I'\in S(v_1)\cup\cdots\cup S(v_k)$ is called the \emph{greedy interval} for $v$.
In case when $v$ has no children the \emph{greedy interval} for $v$ is $[0,\costFunc(v))$.

The algorithm for a rounded down-monotonic cost function simply assigns the greedy interval for each vertex, once such intervals are constructed for the children.
The pseudo-code is shown as Algorithm \ref{alg:down-monotonic}.
\begin{figure}[htb]
\begin{center}
\begin{minipage}{.9\linewidth}
\begin{algorithm}[H]
\SetAlgoRefName{Down-Monotonic}
	\caption{(input: a tree with down-monotonic cost function).}
	\label{alg:down-monotonic}
    Compute a rooted rounding of the input tree.\;
    \For{each vertex $v$ in a bottom-up fashion}
    {
    Compute the greedy interval $I'$ for $v$ and set $f(v)\gets I'$.
    }
    \Return the search strategy that corresponds to $f$.
\end{algorithm}
\end{minipage}
\end{center}
\end{figure}

\subsection{Analysis} \label{sec:down-monotonic-analysis}

In order to prove the correctness and efficiency, we will now introduce concepts that allow to express the extended visibility sequences used in our methodology in terms of the vertex extension operator over integer-valued sequences that were developed for searching in unweighted trees \cite{Onak}.

\begin{definition} \label{def:granularity}
    Let $S(v)$ be an extended visibility sequence for some $v\in V(T)$. Take a $d\in \nat$ that divides $\costFunc(v)$. 
    We say that the visibility sequence $S^{d}(v)$ is \emph{corresponding to} $S(v)$ with \emph{granularity} $d$ if and only if it lists in decreasing order all values obtained in the following way. For every interval $[a,b)\in S(v)$, add to $S^{d}(v)$ the values $\frac{a}{d}+i$ for each $i\in\{0,\ldots,n-1\}$, where $n=\frac{b-a}{d}$.
\end{definition}

Our further analysis will leverage a key property of aligned extended strategy functions, expressed by Observation~\ref{obs:monotonic-granularity}, which is later used in Lemma~\ref{lem:optimal-sequence}.

\begin{observation} \label{obs:monotonic-granularity}
Let $S_1$ and $S_2$ be extended visibility sequences such that there exists a $d\in \nat$ that divides the endpoint of each interval in both $S_1$ and $S_2$. We have that $S_1 <_l S_2$ if and only if $S^d_1 <_l S^d_2$.
\qed
\end{observation}

An extended visibility sequence $S(v)$ of a vertex $v$ is \emph{lexicographically minimal} if:
\begin{itemize}
    \item there exists an extended strategy function $f$ defined on $T_v$ such that $S(v)$ is derived from $f$, and
    \item there exists no extended strategy function $f'$ defined on $T_v$, such that $S'(v)<_l S(v)$, where $S'(v)$ is the extended visibility sequence for $v$ derived from $f'$.
\end{itemize}

The following Lemma~\ref{lem:optimal-sequence} serves as an induction step, which when applied bottom-up over the input tree, shows that Algorithm \ref{alg:down-monotonic} indeed returns an optimal strategy.
\begin{lemma} \label{lem:optimal-sequence}
    Consider a tree $T=(V,E,\costFunc)$, with rounded and down-monotonic $\costFunc$. If the extended visibility sequence of each child of some vertex $v$ is lexicographically minimal, then the greedy interval for $v$ is also lexicographically minimal.
\end{lemma}
\begin{proof}
    Let $v_{1},\ldots,v_{k}$ be the children of $v$.
    Let $f$, defined for the subtrees $T_{v_1},\ldots,T_{v_k}$, be such that for each $i\in\{1,\ldots,k\}$ the lexicographically minimal sequence $S(v_i)$ is derived from $f$. Take extended visibility sequences $\tilde{S}(v_i)$ for each $i\in\{1,\ldots,k\}$ that are derived from an arbitrary extended strategy function.
    By assumption, for each child $v_i$, we have that $S(v_i)\leq _l \tilde{S}(v_i)$. Based on Observation~\ref{obs:monotonic-granularity}, we also have that 
    \begin{equation} \label{eq:Sd}
    S^{d_i}(v_i)\leq_l\tilde{S}^{d_i}(v_i),
    \end{equation}
    where $d_i=\costFunc(v_i)$ for each $i$.
    Note that the $d_i$ may be different for different children $v_i$.
    
    We apply the vertex extension operator (see Definition~\ref{def:vertex_extension}) for $v$ using the sequences $\tilde{S}^{d_i}(v_i)$. Since the vertex extension operator is minimizing \cite{Onak}, it assigns to $v$ the lexicographically minimal sequence $\tilde{S}^d(v)$ available for $v$, out of all valid visibility sequences, given the sequences assigned to the children.
    Note that since $\costFunc$ is rounded and down-monotonic, we have that $d$ divides $d_i$ for each $i\in\{1,\ldots,k\}$.
    For each $i$, we have that \eqref{eq:Sd} implies $S^d(v_i)\leq_l\tilde{S}^d(v_i)$.
    Intuitively, we may switch from the granularities $d_i$ to $d$ preserving the ordering of sequences.
    On the other hand, the vertex extension operator is monotone \cite{Onak} and when applied on the sequences $S^{d_i}(v_i)$, the resulting sequence $S^d(v)$ satisfies $S^d(v)\leq_l \tilde{S}^d(v)$.
    Consider extending $f$ to $T_v$ by assigning to $v$ the greedy interval to $v$.
    By the definition, the extended visibility sequence $S(v)$ derived from such $f$ is such that $S^d(v)$ corresponds to $S(v)$.
    This proves that $S(v)$ is lexicographically minimal as required.
\end{proof}

By induction over the subtrees of $T$, with the induction claim formulated by Lemma~\ref{lem:optimal-sequence}, one obtains that the search strategy returned by Algorithm \ref{alg:down-monotonic} is optimal for the rounding $T$ of the input tree $T'$. 
On the other hand, we have $\OPT{T}\leq 2\cdot\OPT{T'}$.
This implies that the algorithm is $2$-approximate for an arbitrary input tree.
Since the tree has $n$ vertices, the lexicographically minimal sequence for each vertex $v$ has $\bigo(n)$ intervals.
Hence, the computation of the greedy interval takes $\bigo(n)$ time, as well as any other list manipulation task, proving that the algorithm has polynomial running time.
This completes the proof of Theorem~\ref{thm:down-monotonic}.

\section{The algorithm for $k$-monotonic cost functions} \label{sec:algorithm-general-monotonic}

Consider a partitioning of a rounded tree $\inputTree{T}{\costFunc}$ into subtrees $T_1,\ldots,T_l$ as in the definition of a $k$-monotonic cost function (cf. Section~\ref{sec:our-results}).
The subtree $T_i$ that contains the root of $T$ is called the \emph{root subtree}.
In the case of Algorithm \ref{alg:down-monotonic}, our methodology made it easier to frame it as a procedure that is computing an extended strategy function.
This then allows to get the corresponding search algorithm as explained in Section~\ref{sec:preliminaries}.
However, with the tools we have it is more natural to present a procedure (cf. Algorithm \ref{alg:main}) that gives Theorem~\ref{thm:k-monotonic} as an adaptive process that instead of outputting a search strategy, actually performs the adaptive search and returns the target.

\begin{figure}[htb]
\begin{center}
\begin{minipage}{.9\linewidth}
\begin{algorithm}[H]
\SetAlgoRefName{General-Monotonic}
	\caption{(input: a tree $T$ with $k$-monotonic cost function; output: the target).}
	\label{alg:main}
    Let $T'$ be the root subtree of the rounding of the input tree.\;
    Compute a search strategy $\mathcal{A}$ for $T'$.\label{line:cost-strategy}\;
    Execute $\mathcal{A}$. \label{line:executeA}\;
    \If{the target $t$ is in $T'$}
    {
    \Return $t$.
    }
    \Else{
    Let $v$ be the child of a leaf of $T'$ such that $T_v$ contains the target.\label{line:v}\;
    Call recursively Algorithm~\ref{alg:main} with input $T_v$.
    }
\end{algorithm}
\end{minipage}
\end{center}
\end{figure}

To compute the search strategy in line~\ref{line:cost-strategy}, if $T'$ is down-monotonic, use the $2$-approximation Algorithm \ref{alg:down-monotonic}, otherwise use the $8$-approximation Algorithm \ref{alg:up-monotonic}.
If the target $t$ is not in $T'$, then let $v$ be the first vertex outside of $T'$ on the path from the root of $T$ to $t$.
Note that all queries performed within $T'$ are then pointing towards $v$ in their replies.
The same holds during a query to the parent of $v$, which is a leaf of $T'$.
Thus, the vertex $v$ can be correctly selected in line~\ref{line:v}.
The depth of the recursion is at most $k$ and each strategy execution in line~\ref{line:executeA} introduces a cost of at most $8\cdot\OPT{T'}\leq 8\cdot\OPT{T}$.
The recursion correctly ends when $T'$ is such that no vertex of $T'$ has a child that does not belong to $T'$ (i.e., all leaves in $T'$ are also leaves in $T$).
This gives Theorem~\ref{thm:k-monotonic}.

\section{The NP-completeness} \label{sec:npc}

It has been shown in \cite{tree-weighted-CicaleseKLPV16} that the edge variant of binary search is NP-complete already for specific small-diameter trees. We will use this fact to argue that vertex search with $k$-monotonic cost functions is not FPT with respect to $k$.

The edge search problem considers as an input an edge-weighted tree $T'=(V,E,c)$, where $c$ attributes each edge with the cost of the query. For an arbitrary search target $t\in V(T')$, a query to an edge $e\in E(T')$ incurs the cost of $c(e)$ and learns which one of the two components in $T'-e$ contains $t$. For a given input tree, similarly as in the vertex search, the goal of the algorithm is to obtain a search strategy that has minimal total cost in the worst case.
Note that this problem can be equivalently defined as performing vertex search in the line graph of $T'$.

Following the description provided in \cite{DereniowskiKUZ17}, we formulate Lemma~\ref{lem:edge-to-node}, which shows a connection between the two problems.
First, for each tree $T'$ we define its \emph{subdivision} $\inputTree{T}{\costFunc}$ as follows.
Let $M=1+\sum_{e\in E(T')}c(e)$ exceed the weight of all edges in $T'$.
For each $e\in E(T')$ introduce a vertex $v(e)$ with weight $\costFunc(v(e))=c(e)$.
Then, $V(T)=V(T')\cup\{v(e)\colon e\in E(T')\}$, $E(T)=\{\{x,v(e)\},\{v(e),y\}\colon e=\{x,y\}\in E(T')\}$, and $\costFunc(x)=M$ for each $x\in V(T')$.
In other words, we obtain $T$ by putting the weight $M$ on each vertex of $T'$ and subdividing each edge $e$ so that the new vertex $v(e)$ that comes from the subdivision of $e$ `inherits' its weight from $e$. 
Informally, the idea behind this construction is that when performing a vertex search on $T$, one should never query any vertex $x$ in $V(T')$ due to the large weights.
Instead, in order to learn whether $x$ is the target in $T$, it is favorable to query all of its neighbors.
\begin{lemma}[\cite{DereniowskiKUZ17}] \label{lem:edge-to-node}
    For each tree $T'$, the edge search cost of $T'$ equals $\OPT{T}$.
\end{lemma}
\begin{proof}
Consider algorithm $\mathcal{A}$ for vertex search in $T$ based on the edge search algorithm on $T'$.
$\mathcal{A}$ queries only the \emph{old} vertices. If the queried vertex is the target \emph{A} stops. Otherwise, $\mathcal{A}$ continues in the subgraph that can still contain the target, until either the target vertex is queried or the subtree is reduced to a single vertex.
\begin{itemize}
    \item If target is one of the old/original vertices, it will be found by reducing $T$ to a single vertex - the search target. The cost in worst case is OPT(T').
    \item If target vertex is one of the new vertices, $\mathcal{A}$ will continue the search until it queries the search target. Note that every edge of T' is queried when searching for some target vertex in T'. The cost of searching such target in worst case is OPT(T').
\end{itemize}
\end{proof}

It has been shown in \cite{tree-weighted-CicaleseKLPV16} that the edge search is (strongly) NP-complete in the class of spiders of diameter at most $6$.
For each such spider $T'$, the diameter of its subdivision $T$ equals $12$, and the corresponding cost function $\costFunc$ is $4$-monotonic.
This proves Theorem~\ref{thm:npc}.

\section{Conclusions} \label{sec:conclusions}

Our first remark is on a possible complexity gap with respect to $k$. It is not known whether the problem for $1$-monotonic cost functions is NP-complete (cf.~\cite{mfcs-DereniowskiW22}). Hence an open problem remains whether the value of $k=4$ in Theorem~\ref{thm:npc} is in fact tight. We conjecture that the problem is NP-complete for a $k<4$. Our conjecture is supported by the fact that we obtained the constant $k=4$ basically via a `black-box' transformation from the edge search~\cite{tree-weighted-CicaleseKLPV16}. Thus, a more refined reduction might provide the right threshold for $k$.

In a wider context, our work is en route to a search for a constant factor approximation algorithm for an arbitrary cost function in trees. Our choice of looking at $k$-monotonic cost functions may be seen as the problem parametrization with respect to $k$. However, it is interesting if other parametrization methods are also natural and algorithmically promising.

We recall a related tree search problem, called \emph{edge search}, in which one performs queries on edges: each reply provides information which endpoint of the queried edge is closer to the target \cite{Ben-Asher}.
One can similarly define for example an up-monotonic cost function $\omega\colon E(T)\rightarrow\mathbb{R_+}$ for the edge search by requiring that there exists a choice for the root $r$ so that for any two edges $\{x,y\}$ and $\{y,z\}$, if $x$ is closer to $r$ than $y$, then $\omega(\{x,y\})\geq\omega(\{y,z\})$. We have outlined some evidence that for unweighted trees, the edge search has been more challenging to analyze than the vertex search, both when it comes to finding good upper bounds (see \cite{node-rank-linear-Schaffer89} vs \cite{application-db-queries-DereniowskiK06,graphs-Emamjomeh-Zadeh16,MakinoUI01}) and optimal algorithms (see \cite{node-rank-linear-Schaffer89} vs \cite{Onak,edge-rank-linear-LamY98}).
In view of the relation between the two problems in their weighted variants (see Section~\ref{sec:npc}) it may be possible that the above does not carry out to this case.
Hence, a fruitful path to get algorithmic progress for the weighted vertex search problem is to first develop new techniques for the weighted edge search.

Also, we repeat the previously mentioned interesting and challenging open question whether a constant-factor approximation is possible for the vertex search in trees with general weight functions.

\bibliography{references}

\end{document}